\documentclass[a4paper,12pt]{article}
\usepackage{amssymb,amsmath,amsthm}
\usepackage{graphicx}

\usepackage{amsfonts}
\usepackage{latexsym}
\usepackage{color}
\usepackage[english]{babel}
\usepackage{hyperref}
\usepackage{enumitem}
\usepackage{float}

\numberwithin{equation}{section}

\newtheorem{lemma}{Lemma}[section]
\newtheorem{theorem}[lemma]{Theorem}

\newtheorem{corollary}[lemma]{Corollary}
\newtheorem{remark}[lemma]{Remark}

\newtheorem{definition}[lemma]{Definition}
\def\sleq{\lesssim}
\newcommand{\azuno}{{a_r}}
\newcommand{\azdue}{{L}}
\newcommand{\ta}{{\tt a}}

\def\vf{\vphi}

\def\r{m}

\def\im{{\rm i}}
\def\es{{\rm e}}
\def\fke{{\frak e}}

\def\cV{{\mathcal{V}}}
\def\cT{{\mathcal{T}}}
\def\cI{{\mathcal{I}}}
\def\QS{{\mathcal{Q}\kern-0.3pt\mathcal{S}}}
\def\cC{{\mathcal{C}}}

\def\Piz{\mathop\Pi^{\circ}}

\def\tk{{\tt k}}
\def\tN{{\tt N}}
\def\tC{{\tt C}}

\def\cO{{\mathcal{O}}}
\def\cA{{\mathcal{A}}}
\def\cU{{\mathcal{U}}}
\def\cR{{\mathcal{R}}}
\def\cB{{\mathcal{B}}}
\def\ops#1{{OPS}^{#1}}

\def\moyal#1#2{\left\{#1;#2\right\}_M}
\def\poisson#1#2{\left\{#1;#2\right\}}
\def\sm#1{S^{#1,\varsigma}_{C}}
\def\shr#1{S^{#1}_{HR}}
\def\adm#1{ad^M_{#1}}

\newcommand{\ep}{\epsilon}

\newcommand{\R}{\mathbb R}
\newcommand{\C}{\mathbb C}

\newcommand{\Z}{\mathbb Z}

\newcommand{\N}{\mathbb N}
\newcommand{\T}{\mathbb T}

\newcommand{\ii }{{\rm i} }

\newcommand{\vphi}{\varphi}

\newcommand{\csi}{\xi}

\newcommand{\Op}{{Op}^w\,}
\def\op#1{Op^w(#1)}
\def\tk{{\tt k}}
\newcommand{\Snoi}{S_{AN,\delta}}
\newcommand{\OPSnoi}{OPS_{AN,\delta}}
\newcommand{\OPSFnorma}{{\cal OPSF}_{N}}
\newcommand{\Shr}{S_{HR}}

\def\cH{{\mathcal{H}}}
\newcommand{\OPS}{\textrm{OPS}\,}
\newcommand{\varep}{\varepsilon}
\newcommand{\norfou}[3]{ \|#1\|^{(#2)}_{#3} }

\newcommand{\Snorma}{\mathcal{SF}_N}
\newcommand{\pth}{L}
\def\russi{Par08, PS10, PS12}

\definecolor{darkgr}{rgb}{0.0, 0.62, 0.42}
\definecolor{cyan}{rgb}{0.0, 0.72, 0.92}

\begin{document}

\title{{\bf  On the stable eigenvalues of perturbed anharmonic
    oscillators in dimension two}}

\date{}


\author{ Dario Bambusi\footnote{Dipartimento di Matematica, Universit\`a degli Studi di Milano, Via Saldini 50, I-20133
Milano. 
 \textit{Email: } \texttt{dario.bambusi@unimi.it}}, Beatrice Langella\footnote{Dipartimento di Matematica, Universit\`a degli Studi di Milano, Via Saldini 50, I-20133
Milano.
\textit{Email: } \texttt{beatrice.langella@unimi.it}}, Marc Rouveyrol\footnote{
DER de math\'ematiques,
ENS Paris-Saclay,
4 Avenue des Sciences,
91190 Gif-sur-Yvette. \textit{Email: } \texttt{marc.rouveyrol@ens-paris-saclay.fr}}.
 }

\maketitle

\begin{abstract}
We study the {asymptotic behavior of the} spectrum of a quantum system which is a perturbation of a
spherically symmetric anharmonic oscillator in dimension 2. We prove
that a large part of its eigenvalues can be obtained by Bohr-Sommerfeld quantization rule applied to the normal form Hamiltonian
and also admit an asymptotic expansion {at infinity}.
The proof is based on the generalization to the present context of the
normal form approach developed in \cite{BLMnr} (see also \cite{PS10})
for the particular case of $\T^d$.
\end{abstract}
\noindent


\medskip



\tableofcontents
\section{Introduction}\label{intro}

In this paper we study the asymptotic behavior of the spectrum of the
quantum  system 
\begin{align}
  \label{H}
&  H:=H_0+V(x,-\im \nabla)\ ,\quad x\in\R^2\ ,
  \\
\label{H0}
&H_0:=-\frac{1}{2}\Delta+\frac{\|x\|^{2\ell}}{2\ell}\ ,\quad
\ell\in\N^*\ ,\quad \ell\geq 2\ ,
\\
\nonumber
&{\|x\|^{2\ell}:=(x_1^2+x_2^2)^\ell}\,,
\end{align}
where $V(x,-\im \nabla)$ is a self-adjoint pseudo-differential operator
of order smaller than $H_0$. The idea is to consider the system \eqref{H} as a perturbation
of the quantum integrable system \eqref{H0} and to apply the
quantization of normal form theory as developed in \cite{BLMnr} (see
also \cite{\russi}, \cite{BLMres}) in order to obtain an asymptotic
expansion of a large part of the eigenvalues of \eqref{H}.

To describe our result, consider first the integrable system $H_0$; it
was proved by Charbonnel \cite{chaspe,Cha83T} that a large part of its
spectrum can be constructed through Bohr-Sommerfeld
rule. Precisely there exists a function $\tilde h_0$ with the property
that $ \tilde h_0(\ta)$ is asymptotic to an eigenvalue of $H_0$ for all $\ta$'s in a subset
$\tC$ of the lattice $\Z^2+{{\kappa}}$ with
${{\kappa}}\in\R^2$ fixed. The function $\tilde h_0$ is a
perturbation of the classical Hamiltonian $h_0$ written in terms of
action angle variables.

Here we prove that there exists a subset $\Omega\subset\tC$
of density 1 in $\tC$ {and a function $\tilde h$ s.t.}
for all $\ta\in\Omega$ there exists an
 eigenvalue $\lambda_{\ta}$ of \eqref{H} {with}
$$
\lambda_{\ta}\sim \tilde h(\ta)\ .
$$
Furthermore, $\tilde h$ is a perturbation of $\tilde h_0$ and admits an asymptotic expansion in powers of
$\|a\|^{-1}$. 

{Results of this kind are completely standard in the case of quantum
oscillators in dimension 1, but we are not aware of
a single result of this kind in dimension at least 2 (see however
\cite{roy,HSV} for related results).}

The present paper can be considered as a continuation of the works
\cite{BLMnr,BLMres,BLMunbdd}, in particular of \cite{BLMnr}, in which
we studied the stable eigenvalues of Schr\"odinger operators on
$\T^d$. The idea of \cite{BLMnr} consists in working on the symbol of
the quantum operator in order to conjugate $H$ to a system in
``quantum normal form''. Now, in classical Hamiltonian mechanics, it
is well known that the normal form of a system is particularly simple
only in the regions of the phase space where the frequencies are non
resonant. A similar property holds also for the quantum system and it turns
out that the eigenvalues which correspond to the classically nonresonant region
can be obtained by Bohr-Sommerfeld quantization rule, and furthermore
admit an asymptotic expansion in powers of $\lambda^{-1}$.

The starting point
of the present paper is the idea that all the techniques developed
in \cite{BLMnr,BLMres,BLMunbdd,PS10} should apply to any quantum
system which is a perturbation of an integrable quantum
system. However there are several technical difficulties to overcome
in order to actually transform such a heuristic statement into a
theorem and, in order to avoid the risk of getting empty
results, we decided to start our investigation from a concrete model
which has some interest in itself, namely the quantum anharmonic
oscillator. The main aim of this paper is to develop the tools needed
to apply the ideas of \cite{BLMnr,BLMres,BLMunbdd} to general
perturbations of quantum integrable systems, and to start to deduce
some consequences.

The main technical difficulties with respect to \cite{BLMnr} are
of four kinds:
\begin{itemize}
\item[(1)] find a class of symbols suitable for the normal form
construction and  deal with it
\item[(2)] generalize
the construction of \cite{BLMnr} to the case where the frequencies do
not coincide with the actions, but are just a function of the actions
which is poorly known
\item[(3)] verify that all the needed nonresonance
properties are {satisfied} in the concrete case of the quantum
anharmonic oscillator
\item[(4)] use a suitable version of functional
calculus in order to obtain spectral properties of operators in
quantum normal form.
\end{itemize}
Concerning (1), we use a class of symbols which is a small modification
of the class used by Helffer and Robert in \cite{HR82,hero2}. However, it
turns out that such a
class of symbols behaves badly under Fourier expansion, so we
use a trick from classical normal form theory in order to avoid to
re-expand symbols in Fourier series at each step of the
iteration. Concerning (2), we use here the remarkable fact that the
actions turn out to be quasi-homogeneous functions of the phase space
variables. Furthermore the Hamiltonian and the frequencies turn out to
be homogeneous
functions of the actions. Concerning (3), we first prove that the actions
are analytic functions of the phase space variables which are globally
defined. We do it by a
direct computation, but the computation is enlightened by the general
theory of action angle variables in the form of \cite{BF}. Then we
have to verify that the subset of the phase space in which the
frequencies are nonresonant has large measure. This is done using
tools from degenerate KAM theory (see \cite{Rus01,BBM11}) and homogeneity
of the frequency map. We also use some results from \cite{Fej04,BF18} (see
also \cite{BFS18}). The point (4) is solved using the same
ideas developed by \cite{chaspe}.

As we mentioned above, our technique is a generalization of a
technique introduced on $\T^d$, so we conclude this introduction by
recalling the results which are known in that situation and that we
hope to extend to more general situations in the future.

In the case of the Schr\"odinger operator on $\T^d$ it was shown by
\cite{FKT,Fri} (see also \cite{Vel,Kar96,PS10}) that {most of the
eigenvalues $\lambda$ of the Laplacian are stable under perturbation, in the sense that all eigenvalues bifurcating from them admit a full asymptotic
expansion in powers of $\lambda^{-1}$.} On the contrary (see
\cite{FKT2}) some eigenvalues are unstable (in the terminology of
\cite{FKT2}), and in particular {they do not admit such an asymptotic
expansion.} In \cite{BLMnr} the stable eigenvalues were recognized to
correspond to the nonresonant regions of the classical phase space,
and in \cite{BLMres} we proved that the unstable eigenvalues can be
obtained as eigenvalues of a Schr\"odinger operator on a lower
dimensional torus. Finally in \cite{BLMunbdd} we used the results of
\cite{BLMres} to study the case of time dependent potentials and to
prove a $\langle t\rangle^\epsilon$ estimate on the growth of Sobolev
norms of the solutions of the time dependent Schr\"odinger
equation. {We plan to investigate analogous problems {for the
    anharmonic oscillator} in the future.}

\noindent{\it Acknowledgments} First of all we thank Didier Robert for
pointing to our attention the papers by Charbonnel, then we thank
Francesco Fass\`o for pointing to our attention the paper \cite{BF}
and San V\~{u} Ng\d{o}c for several suggestions on quantum action angle
coordinates. During the preparation of this work we also had several
discussions with Alberto Maspero, that we warmly thank. 

We thank the Italian Gruppo Nazionale di Fisica Matematica of INDAM
for the support.

\section{Main result}\label{main}

\subsection{Symbols} \label{syb.1}

As usual, we define a scale of Sobolev like spaces adapted to our situation.
For all $s \geq 0$ we define $\cH^s = D(H_0^{s\frac{\ell+1}{2\ell}})$
(where $D(\cdot)$ is the domain of an operator), while if $s<0$ we
define $\cH^s = \left(\cH^{-s}\right)^\prime$, where $V^\prime$ is the
dual space to $V$ with respect to $\cH^0 =L^2(\R^2)$. We consider such
spaces endowed with the natural norms $\| \psi \|_s := \|
H_0^{s\frac{\ell+1}{2\ell}} \psi\|_{L^2}\,$.

We will denote by $\cB \left(\cH^s, \cH^{s'}\right)$ the space of
bounded linear operators from $\cH^s$ to $\cH^{s'}$, and, given an operator
$A\in \cB \left(\cH^s, \cH^{s'}\right)$, we will denote by 
$\| A\|_{s, s'}$ its norm in $\cB \left(\cH^s, \cH^{s'}\right) $.

\begin{definition}\label{smooth.def}
	Given {$N \in \R^+$} we say that an operator $R$ is smoothing of order
        $N$ if $R \in \bigcap_{s \in \R} \cB \left(\cH^s, \cH^{s +
          N}\right)$.
        \\ If $R$ is smoothing of order
        $N$ for any {$N \in \R^+,$} we say $R$ is an (infinitely)
        smoothing operator.
\end{definition}

Following \cite{HR82}, we define now a first class of symbols. Denote
$$\tk_0(x,\xi) :=
(1+\|x\|^{2\ell}+\|\xi\|^{2})^{\frac{\ell+1}{2\ell}} \ ,\quad
\|x\|:=\sqrt{x_1^2+x_2^2}\ , $$
\begin{definition}
	\label{symbol.ao}
	A function $f=f(x,\xi)$ will be called a symbol of order  $m\in\R$ if  $f \in C^\infty(\R^4)$ and 
	$\forall \alpha, \beta \in \N^2$, there exists $C_{\alpha, \beta} >0$ s.t. 
	\begin{equation}
	\label{es.6}
	\vert \partial_\xi^\alpha \, \partial_x^\beta f(
	x,\xi)\vert \leq C_{\alpha,\beta}
	\ \tk_0(x,\xi)^{m-\frac{|\beta| +\ell |\alpha|}{\ell+1}}  \ ,
	\end{equation}
 {where $\left|\alpha\right|:=|\alpha_1|+|\alpha_2|$, and similarly for $|\beta|$.} In this case we will write $f \in S^\r_{HR}$.
\end{definition}

\begin{definition}\label{pseudo.an}
	To any symbol $f \in \Shr^m$, $m \in \R$, we associate an
        operator $\Op(f)$ acting on the scale $\{\cH^s\}_{s \in \R}$
        as the Weyl quantization of $f$, namely:
	\begin{equation}
	[\Op(f) \psi](x) = \frac{1}{4\pi^2}\int_{\R^2}\int_{\R^2}
        f\left(\frac{x + y}{2}, \xi\right) \psi(y) e^{\im (x -
          y)\xi}\ d y\ d \csi \ .
	\end{equation}
	Conversely, if there exist a symbol $f \in \Shr^{m}$ s.t. 
	$$
	F = \Op(f)\,,
	$$
	we say that $F$ is a pseudo-differential operator of order
        $m,$ and write $F \in \OPS_{HR}^{m}\,.$
\end{definition}

We also need to deal with symbols which are function of the actions
only.

\begin{definition}\label{simboli S m delta}
Given $\varsigma > 0$, $m \in \R$, we define
the class of classical symbols $\sm m$ as the set of all the functions
$f \in \C^\infty(\R^2)$ such that for any $\alpha \in \N^2$,
there exists $C_{\alpha}$ s.t.
\begin{equation}
  \label{semi}
|\partial_a^\alpha g(a)| \leq C_{\alpha}
\langle a \rangle^{m - \varsigma |\alpha|}\,,
\end{equation}
where, as usual, $\langle a\rangle:=\sqrt{1+\|a\|^2}$.
\end{definition}

\begin{definition}
\label{sim.1}
Given a sequence of symbols $\left\{f_j\right\}_{j\geq 0}$ with
$f_j\in \sm{m-\rho j}$ for some $m\in\R$ and $\rho, \varsigma>0$, and a function
$f(a)$,  we
write 
\begin{equation}
\label{sim.eq}
f\sim \sum_{j}f_j\ ,
\end{equation} 
if for any $N\in \N$ there exists $C_N$ s.t.
\begin{equation}
\label{sim.eq.2}
\left|f(a)-\sum_{j=0}^{N}f_j(a)\right|\leq{
  C_N}{\langle a\rangle^{m-(N+1)\rho}}\ .
\end{equation}
\end{definition}

We will use similar notations for all different classes of symbols we
will meet in the following. 
{We finally fix some further notation: given two quantities $a,
  b \in \R$, we write $a \lesssim b$ if there exists a positive
  constant $C$, independent of all the relevant quantities, such that
  $a\leq C b$. {We will occasionally write $a \lesssim_{s} b$ if the constant $C$ depends on the parameter $s$.}} {We will also write $a\simeq
b$ if $a\sleq b$ and $b\sleq a$.}

\subsection{The integrable case}\label{integrabile}

We present here the results of Charbonnel \cite{chaspe} on the
spectrum of $H_0$, in a form suitable for our developments.

Consider the classical Hamiltonian system
\begin{equation}
  \label{h0}
{h_0(x,\xi):=\frac{\left\|\xi\right\|^2}{2}+\frac{\|x\|^{2\ell}}{2\ell}}\ .
\end{equation}
whose quantization is  $H_0$. 
We introduce now action variables $a_1,a_2$
for $h_0$; it turns out (see Lemma \ref{azioni} below) that their
range is the cone
\begin{equation}
  \label{pigreco}
  \Pi:=\left\{a\in\R^2\ ;\ a_1\geq 0\ \textrm{ if } a_2 \geq0\,,\quad  a_1 \geq |a_2|\ \textrm{ if } \ a_2< 0 \right\}\,;
\end{equation}
{\bf we fix once for all an open cone $\cC$} such that $\overline{\cC}\setminus \{0\}$
is contained in the interior of $\Pi$. 

Consider now the operators $A_1,$ $A_2$ obtained by Weyl
quantization of the actions $a_1,$ $a_2$. Since $[A_1;A_2]=0$ (as it immediately
follows from the fact that $a_2$ is the angular momentum, which is a quadratic polynomial in
$x,\xi$), one can consider their joint spectrum. 
Precisely
there exist two diverging sequences
$\lambda_{\ta_1}^{(1)}$ and $\lambda_{\ta_2}^{(2)}$ and a basis of $L^2$
formed by 
joint eigenfunctions $\psi_{\ta}$, $\ta\equiv(\ta_1,\ta_2)$:
\begin{align}
  \label{eige}
A_1\psi_{\ta}=\lambda^{(1)}_{\ta_1}\psi_{\ta}\ ,\quad
A_2\psi_{\ta}=\lambda^{(2)}_{\ta_2}\psi_{\ta}\ .
\end{align}
Then
\begin{equation}
  \label{spet}
\Lambda^{A}:=\left\{
(\lambda_{\ta_1}^{(1)},\lambda_{\ta_2}^{(2)} )\right\} 
\end{equation}
 is called the \emph{joint spectrum of $A_1$ and
  $A_2$}. The following theorem is essentially Theorem 2.4 of
\cite{chaspe}
\begin{theorem}
  \label{cha.1}
There exist ${{\kappa}}\in\R^2$ and $C_0$ with the following properties:
  \begin{equation}
    \label{cha.1.1}
\Lambda^A\cap\cC\subset
\bigcup_{\ta\in\Z^2+{{\kappa}}} B_{\left(\frac{C_0}{\|\ta\|}\right)}(\ta)\ , 
  \end{equation}
  where $B_{R}(\ta)$ is the closed ball in $\R^2$ of radius $R$ and center
  $\ta$.

  Furthermore, for $\ta\in(\Z^2+\kappa)\cap\cC $ large enough
  \begin{equation}
    \label{cha.1.2}
\sharp \left(\Lambda^A\cap B_{\left(\frac{C_0}{\|\ta\|}\right)}\left(\ta\right)\right)=1\ .
  \end{equation}
\end{theorem}

Actually, from functional calculus (see \cite{Cha83}) one can deduce
an asymptotic expansion of the eigenvalues. 

\begin{theorem}
  \label{cha.2}
  Let
  \begin{equation}\label{reticolo.cono}
  \tC:=(\Z^2+{{\kappa}})\cap\cC\,,
  \end{equation}
  with $\kappa$ as in Theorem \ref{cha.1}. There exists a sequence
  $\tilde h_{0,j}\in S_{C}^{\frac{2\ell}{\ell+1}-j, 1}$, $j\geq 0$, 
  of
  symbols with the following property: for any
  ${\ta\in\tC}$ large enough, there exists a unique
  eigenvalue $\lambda^{(0)}_{\ta}$ of $H_0$ fulfilling
\begin{equation}
  \label{cha.2.1}
\lambda_{\ta}^{(0)}{=} \tilde h_0(\ta)\sim\sum_{j\geq 0}\tilde
h_{0,j}(\ta)\ .
\end{equation}
Furthermore $\tilde h_{0,0}$ is the Hamiltonian $h_0$ written in action angle
variables. 
\end{theorem}

\subsection{Main result}

We are now ready to state our main result, which is the following
Theorem. 

\begin{theorem}
\label{maint}
Consider the operator
\begin{equation}
  \label{ope1}
H:=H_0+V \ ,
\end{equation}
with $V\in\ops m_{HR}$ and
\begin{equation}
  \label{m}
m<\frac{2\ell}{\ell+1}\,;
\end{equation}
 define
\begin{equation}\label{def.M}
	M:= \frac{\ell-1}{\ell + 1}\,,\quad
        \fke:=\frac{2\ell}{\ell+1}-m\ ,\quad \delta_0:=M - \min
        \left\{\frac{1}{\ell
          +1},\,\frac{\fke}{3},\,\frac{2}{7}\right\}  ,
\end{equation}
then there exists $\mu_0>0$ such that for any choice
of the parameters $\delta$ and $\ep$ satisfying 
\begin{equation}\label{def.d.ep}
\delta_0 < \delta<M\,, \quad 0< \ep <\frac{M-\delta}{2 \mu_0} \,,
\end{equation}
define
\begin{equation}\label{def.varsigma}
\varsigma := 1 - (M-\delta)
\ ,\quad 
\rho:=\min\left\{\fke-3(M-\delta),2-7(M-\delta)\right\}\,,
\end{equation}
then the following holds.

There exists a sequence of symbols $\{\tilde z_j\}_{j \in \N}$ with
$\tilde z_j\in \sm{m-j\rho
}$ and
a set $\Omega \subset \tC$ such that
\begin{enumerate}
\item $\Omega$ has density one {at infinity} in $\tC$, more precisely, {denoting $\forall R>0$ $B_R:= B_R(0)$,} one has
\begin{equation}
\label{density}
1-\frac{\sharp(\Omega      \cap B_R)}{\sharp(B_R \cap   \tC 
	)} = \cO\left(R^{-\frac{(M-\delta)}{\mu_0} + 2\ep}\right)\quad \textrm{as } R \rightarrow + \infty\,;
\end{equation}
\item for any
$\ta\in \Omega$ there exists an eigenvalue $\lambda_{\ta}$
of \eqref{ope1}
which admits the asymptotic expansion
\begin{equation}
\label{asym}
\lambda_{\ta}\sim  \tilde h_0(\ta)+\sum_{j\geq0}\tilde z_j(\ta)\ ,\quad \ta\in
\Omega \ ,
\end{equation}
where $\tilde{h}_0(\ta)$ is the function in \eqref{cha.2.1}.
\end{enumerate}
\end{theorem}

\begin{remark}
  \label{esempio}
An example of a perturbation fulfilling the assumptions of Theorem
\ref{maint} is the Weyl quantization of
\begin{equation}
  \label{esem.1}
  v(x,\xi):=\sum_{|\alpha|+\ell
    |\beta|<2\ell}c_{\alpha,\beta}x^{\alpha}\xi^\beta \ ,
\end{equation}
where
$$
x^{\alpha}:=x_1^{\alpha_1}x_2^{\alpha_2}\ ,\quad
\xi^{\beta}:=\xi_1^{\beta_1}\xi_2^{\beta_2} \ . 
$$
\end{remark}

The rest of the paper is devoted to the proof of Theorem \ref{maint}.

\section{Scheme of the proof}\label{scheme}

The idea of the proof is exactly the same of \cite{BLMnr} (see also
\cite{PS10}). We now recall it in order to give a road map to the
reader. The idea is to perform a ``semiclassical normal form''.

Consider first
the classical Hamiltonian
\begin{equation}
  \label{sturm.cl}
h:=h_0+V\ ,
\end{equation}
with $h_0$ integrable. As already anticipated we denote by
$a=(a_1,a_2)$ the classical actions of $h_0$. Suppose for a while to introduce
action angle variables, so that $h_0$ turns out to be a function of
the actions only. 

We are interested in studying the system in the region of large $a$.  In such a
region, $V$ can be considered as a perturbation of $h_0$.  Thus we develop a
perturbation theory in which the order of perturbation is given by
inverse powers of $\left\|a\right\|$.

The classical normal form procedure consists in looking for an
auxiliary Hamiltonian function $g$ s.t. the corresponding time 1 flow
$\phi_g^1$ (namely the time one flow of the corresponding Hamiltonian
vector field) conjugates $h$ to a new Hamiltonian $h\circ\phi^1_g$ which,
up to lower orders, is a function of the actions only.

By a formal computation one has (see below)
$$ h\circ\phi_g=h_0+\{g,h_0\}+V+{\rm lower\ order\ terms}\ ,
$$ where $ \left\{\cdot,\cdot\right\} $ are the Poisson brackets. Then, the
main point is to determine $g$ in such a way that
\begin{equation}
  \label{hom}
\{g,h_0\}+V= {\rm function \ of} \ a\ {\rm only}\ .
\end{equation}
However, this can be done \emph{only in the nonresonant regions}. To
explain the situation {pass to action angle variables $(a,\vf)$ and expand $V$
  in Fourier series in the angles}: one
has
$$
V(a,\vphi)=\sum_{k\in\Z^2}V_k(a)e^{\im k\cdot \vphi}\ ,
$$
and, defining  the frequencies as
$$
\omega_i(a):=\frac{\partial h_0}{\partial a_i}(a)\ ,
$$
one gets
$$
\left\{ h_0; \cdot\right\}=-\omega\cdot \frac{\partial}{\partial \vf}\ ,
$$
so that one is led to try to define the function $g$ as
$$
\sum_{k\not=0}\frac{V_k(a)}{\im \omega(a)\cdot k}e^{\im k\cdot \vphi}\,,
$$ which of course is ill defined in the resonant region where the
denominators vanish. To overcome this problem we introduce a cutoff to
localize outside the resonant regions. It turns out that a suitable
cutoff can be defined as follows: fix once for all a function $\chi\in
C^\infty(\R)$ which is equal to 1 in $B_{1/2}(0)$ and vanishes outside
$B_1(0)$, define
\begin{equation}
  \label{solhomo.0}
\chi_k\left(a\right):=\chi\left(\frac{\omega\cdot
  k}{\|k\|\|a\|^{\delta}}\right)\ ,\quad \widetilde
\chi_k\left(a\right):=\chi \left(\frac{\|k\|}{\left\|a\right\|^\epsilon}
\right)\ , 
\end{equation}
and put
\begin{equation}
  \label{gquasivera}
g(a,\vphi):=\sum_{k\not=0}\frac{V_k(a)}{\im \omega(a)\cdot k}\left(1-\chi_k(a)\right)\widetilde\chi_k(a)e^{\im
  k\cdot \vphi}\ ,
\end{equation}
which can be used to play the game it was designed for. By doing this
and iterating the construction one conjugates the system to a normal
form part, which in the nonresonant regions depends only on the
actions plus a remainder which decreases at infinity as fast as one
wants.

Now the point is to quantize this procedure.  This was done in
\cite{BLMnr} for the case of the Schr\"odinger operator on the torus
$\T^d$. The key remark is that, by Egorov theorem, one can simply
quantize the transformation given by the Hamiltonian flow of $g$ and
this gives a unitary transformation that conjugates $H$ to a ``quantum
normal form''\null
\footnote{Of course there are many technical
  details to verify, but this will be done in the forthcoming sections}. However
the situation of \cite{BLMnr} (and also of \cite{\russi}) was quite
simplified by the fact that the Laplacian on $\T^d$ is the
quantization of $h_0:=\sum \xi_j^2$, which is already in action angle
variables and has the remarkable property that the frequency map is
very simple.

Here, in order to keep at a minimum level the technicalities, we work
without using explicitly action angle variables and we make all
the developments in the original Cartesian coordinates $(x,\xi)$. A
priori, the main difficulty in doing this consists in solving the
homological equation \eqref{hom}: to this end we use here a method developed
in \cite{Bam96} which consists in making Fourier developments based on
the Hamiltonian flow of the actions. This requires a careful study of
the properties of the flow of the actions variables and of the
behavior of symbols under such a flow: it will be done in
Subsects. \ref{foa} and \ref{expans}.

As anticipated in the Introduction, the other difficulties are mainly
related to the study of the structure of the actions and of the frequency
map.

\section{Properties of $h_0$}\label{action1}

In this section we study the properties of the actions variables of
$h_0$. In this study, a relevant role is played by quasi-homogeneous functions. 
\begin{definition}
  \label{quasihom}
  A function $f\in C^\infty(\R^{4}\setminus \left\{0\right\})$ is said
  to be quasi-homogeneous of degree $m$ if for any
  $(x,\xi)\in\R^4\setminus\left\{0\right\}$ one has
   $$
f(\lambda x,\lambda^\ell \xi)=\lambda^m f(x,\xi)\ ,\quad \forall
\lambda>0\ .
$$
\end{definition}
In the following we will also use functions (of the actions) which are
homogeneous in the standard sense.

 We remark that if $f$ is quasi-homogeneous of degree $m$ then
$\partial_x^\alpha\partial_\xi^\beta f$ is quasi-homogeneous of degree
{$m-|\alpha|-|\beta|\ell$. } 

\begin{remark}
  \label{qosmooth}
If $f$ is quasi-homogeneous of some degree $m$ and smooth \emph{over the whole $\R^{4}$},
then it is also a symbol of class $\shr {\frac{m}{\ell+1}}$. 
\end{remark}

The main limitation of the above remark is that a quasi-homogeneous
function is $C^\infty$ until the origin only if it is a polynomial.
Nevertheless, since we are interested in the behavior at infinity of
the symbols, this is not a problem. 
To make this precise, {we \emph{fix once for all a cutoff function $\chi$} which is {even and}
$C^\infty(\R,\R)$, and is supported in $[-1,1]$ and equal to 1 in
$[-\frac{1}{2},\frac{1}{2}]$.}  
\begin{definition}
  \label{cutoffed}
Given a function $f:\R^4\to\R$ we define  
\begin{equation}
  \label{cutoofed}
f_{\chi}(x,\xi):= (1-\chi(\|a(x,\xi)\|))f(x,\xi)\ .
\end{equation}
\end{definition}

\begin{remark}
  \label{cuto}
If $f$ is quasi-homogeneous of degree $m$ then $f_\chi\in
\shr{\frac{m}{\ell+1}}$.
\end{remark}

With a slight abuse of notation, {\bf in this Section \ref{action1}
   we will say
  that $f\in\Shr ^m$ if $f_{\chi}\in \Shr^ m$.} Analogously, {\bf in Section \ref{psidoo}} we will define a new class of symbols, $\Snoi^m$, and {\bf we will say that $f \in \Snoi^m$ if $f_\chi \in \Snoi^m$.}

During the rest of the paper however we will work carefully because
the functions to be quantized are everywhere defined, so we have to
take into account their behavior on the whole of $\R^4$. So starting
from  Section \ref{nfl} we will come back to the correct terminology.

We start now the study of the action variables for $h_0$. Their
properties can be deduced from the general theory of integrable
systems (we follow here the ideas of {\cite{BF,CdV2})}. Here we give
a direct proof of all the properties working on the Hamiltonian
$h_0$. We recall that the actions, by their definition have the
property that the flow they generate is periodic in time with period
$2\pi$. 

The action variables can be
defined to be the angular momentum
\begin{equation}
  \label{L}
a_2(x,\xi):=x_1\xi_2-x_2\xi_1\ 
\end{equation}
and the radial action, namely the action of the effective Hamiltonian
\begin{equation}
  \label{efficace}
h_0^*(r,p_r,L):=\frac{p_r^2}{2}+V_{L}^*(r)\ ,\quad
V_{L}^*(r):=\frac{L^2}{2r^2}+\frac{r^{2\ell}}{2\ell}\ ,
\end{equation}
where $r={\|x\|}$ and $p_r$ is the conjugated
momentum. By $L$ we mean here the value taken by the angular momentum,
namely we mean that we are on a level surface {$a_2(x,\xi)=L$}.

In order to define $a_1$, \emph{for $L\not=0$} we preliminary define
\begin{equation}
  \label{ar}
a_r=a_r(E,L):=\frac{\sqrt2}{\pi} \int_{r_m}^{r_M}\sqrt{ E- V^*_L(r)}dr\ ,
\end{equation}
where $0<r_m <r_M$ are the two solutions of the equation
$$
E-V^*_L(r)=0\ .
$$

\begin{lemma}
  \label{azioni}
  For $L>0$, consider the function
  \begin{equation}
    \label{a1}
a_1(E,L):=a_r(E,L)\ ,\quad L>0\ .
  \end{equation}
Such function has the following properties:
  \begin{itemize}
  \item[(1)]  it extends to a complex analytic function (still denoted
    by $a_1$) of $L$ and $E$
    in the region
    \begin{equation}
      \label{defia1}
\left|L\right|<\left(\frac{2\ell}{\ell+1}E\right)^{\frac{\ell+1}{2\ell}}\ ,\quad
E>0\ ;
    \end{equation}
  \item[(2)] for $L<0$ one has
    \begin{equation}
      \label{lpiuomeno}
a_r(E,L)=a_1(E,L)+L\ ;
    \end{equation}
  \item[(3)] the function $a_1(x,\xi):=a_1(h_0(x,\xi),a_2(x,\xi))$ is
    quasi-homogeneous of degree $\ell+1$, so in particular it is
   of class $C^1(\R^4)$;

\item[(4)] the map $E\mapsto a_1(E,a_2)$ admits an inverse
  $E=h_0(a_1,a_2)$ which is analytic in the interior of $\Pi$. Furthermore it is homogeneous of degree
  $\frac{2\ell}{\ell+1}$, namely one has
  $$
h_0(\lambda a_1,\lambda a_2)=\lambda^{\frac{2\ell}{\ell+1} }h_0(a_1,
a_2)\ , \quad \forall \lambda>0\,.
  $$
  \end{itemize}
\end{lemma}
The proof of this Lemma is postponed to Appendix \ref{sec.residui}. 
\begin{corollary}
  \label{azioniefrequenze}
 Denote  $\omega:=(\omega_1,\omega_2)$, with
  $\omega_i(a):=\frac{\partial h_0}{\partial a_i}(a)$, then $\omega_i$
  are homogeneous of degree
  \begin{equation}
    \label{M}
M:=\frac{\ell-1}{\ell+1}\ 
  \end{equation}
  as functions of $a$. Furthermore
  $\omega_i(x,\xi):=\omega_i(a_1(x,\xi),a_2(x,\xi))$ is quasi-homogeneous
  of degree $\ell-1$.
\end{corollary}

\begin{remark}
  \label{equia}
{One has that} $a_i$ are quasi-homogeneous functions of order
$\ell+1$, and furthermore $a=0$ implies $x=\xi=0$. 
It follows that 
$$
\langle a(x,\xi)\rangle \simeq \tk_0(x,\xi) \ .
$$
\end{remark}

\section{Symbols, Fourier expansions and Pseudo-differential calculus}\label{psidoo}

\subsection{A class of symbols}
In the following we will need a class of symbols slightly more general than those
of Definition \ref{symbol.ao}. 

First of all we fix positive parameters $\ep, \delta$ fulfilling
\begin{equation}
\label{delta}
\frac{\ell-2}{\ell+1}<\delta<\frac{\ell-1}{\ell+1}=M\,, \quad 0<\ep < M -\delta;
\end{equation}
further requirements on $\delta$ and $\ep$ will be specified later on
{(see equation \eqref{legami.tra.parametri.1} and Remark \eqref{mu.legame} below).} 
We also define
\begin{equation}\label{delte}
\delta_1:=\delta- \frac{\ell-2}{\ell+1} \ ,\quad \delta_2:=\delta_1+
\frac{\ell-1}{\ell+1}\ , 
\end{equation}
and remark that
\begin{equation}
\label{ridelte}
\delta_1+\delta_2=2\delta-\frac{\ell-3}{\ell+1}\,\ \Longrightarrow\quad \frac{\ell-1}{\ell+1}<\delta_1+\delta_2<1 \ .
\end{equation}

\begin{definition}
	\label{symbol.a1}
	Given $f \in C^\infty(\R^4)$, we will write $f \in \Snoi^{m}$ if $\forall \alpha, \beta \in \N^2$, there exists $C_{\alpha, \beta} >0$
	s.t.
	\begin{equation}
	\label{es.7}
	\vert \partial_x^\alpha \, \partial_\xi^\beta f(
	x,\xi)\vert \leq C_{\alpha,\beta}
	\ \left(\tk_0(x,\xi)\right)^{m-\delta_1{|\alpha|}-\delta_2
		{|\beta}|}  \ ,
	\end{equation}
	with $\delta_1, \delta_2$ given by \eqref{delte}.
{We will say that an operator $F$ is a
pseudo-differential operator of class $\OPSnoi^{m}$ if there exists symbol $f \in \Snoi^{m}$ s.t. $F=\Op (f)$.}
\end{definition}

The smallest constants $C_{\alpha,\beta}$ s.t. Eq. \eqref{es.7} holds
	form a family of seminorms for the symbols of this class. However
        in order to get the standard algebra properties it is more
        convenient to use a different definition.
        \begin{definition}
          \label{seminorme}
          Let $f\in\Snoi^m$, then, $\forall \alpha,\beta\in\N^d$, we put
          \begin{equation}
            \label{norma.k}
	\norfou{{f}}{m}{\alpha, \beta}:= \sup_{x, \xi}
        \sup_{|\alpha^\prime| \leq |\alpha|, |\beta^\prime| \leq
          |\beta|} \left| \partial_{x}^{{\alpha^\prime}} \partial_\xi^{{\beta^{\prime}}}
            {f}(x, \csi)\right| (\tk_0(x, \csi))^{-(m-|{\alpha^\prime}|\delta_1
              - |{\beta^\prime}|\delta_2)}\,.
          \end{equation}
        \end{definition}
    \begin{remark}
    	\label{prod 2}
    	With this definition we have
    	$$ \norfou{fg}{m+m'}{\alpha,\beta}\leq
    	\norfou{f}{m}{\alpha,\beta}\norfou{g}{m'}{\alpha,\beta}
    	$$
    	for any couple of symbols $f\in\Snoi^m$, $g\in\Snoi^{m'}$. 
    \end{remark}

    \begin{lemma}\label{lemma composizione 2}
    	Let $m, m' \in \R$,  $F = \op f \in \OPSnoi^{m}$, $G =
    	\op g \in \OPSnoi^{m'} $. Then  $FG \in \OPSnoi^{m +
    		m'}$. Denote by $f\sharp g$ its symbol, then it admits the asymptotic expansion 
    	\begin{align}
    	\label{sigma1 2}
    	& f\sharp g\sim \sum_{j \geq 0} (f\sharp g)_j\ , \\
    	\label{sigma2 2}
    	& (f\sharp g)_j:= \frac{1}{\ii^j }\sum_{|\alpha| + |\beta| = j} \Big(
    	\frac12 \Big)^{|\alpha|} \Big(- \frac12 \Big)^{|\beta|}
    	(\partial_x^\beta \partial_\xi^\alpha f )( \partial_x^\alpha
    	\partial_\xi^\beta g ) \in \Snoi^{m + m' - (\delta_1+\delta_2) j}, \quad j \geq 0\,.
    	\end{align}
    	Furthermore, $\forall n, \alpha,\beta$, $\exists \alpha',\beta'$ and $C>0$ s.t.
    	\begin{equation}
    	\label{stime2}
    	\left\|f\sharp g-\sum_{j=1}^{n-1}(f\sharp
    	g)_j\right\|^{(m-m'-n(\delta_1+\delta_2))}
    	_{\alpha,\beta}\leq C\left\| f\right\|^{m}_{{\alpha',\beta'} }\left\|
    	g\right\|^{m'}_{{\alpha',\beta'}}\ .
    	\end{equation}
    \end{lemma}
   \begin{lemma}
   	\label{calderon}(Calderon Vaillancourt theorem)
   	Let $f\in\Snoi^m$, then $\forall s\in\R$, $\Op(f)\in
        B(\cH^{s+m};\cH^{s})$ and there exist $ C,\alpha,\beta$,
        s.t.
   	$$
   	\left\| \Op(f)\right\|_{{s+m}, {s}}\leq C
   	\left\|f\right\|^{(m)}_{\alpha,\beta}\ .
   	$$
   \end{lemma}

\subsection{Flow of the actions}\label{foa}

In order to develop perturbation theory, we will make use of a Fourier
development based on the flow of the actions (following the ideas of
\cite{Bam96}). For this reason such a flow has to be studied quite in detail.

Let $a_i$ be the $i$-th action. Consider the corresponding
Hamilton equations, namely
\begin{equation}
  \label{az.eq}
\dot\xi=-\frac{\partial a_i}{\partial x}\ ,\quad \dot x=\frac{\partial a_i}{\partial \xi}\ 
\end{equation}
and denote by $\phi_{a_i}^{\vf}$ the corresponding time $\vf$
flow.

We are now going to study $\phi^\vf_{a_i}$ and, for $f\in\Snoi^m$,
$f\circ\phi_{a_i}^\vf$. We concentrate on the non trivial action
$a_1$, but all what follows is only based on the property that $a_1$
is quasi-homogeneous of degree $\ell+1$, so it trivially holds also for $a_2$.
We will also distinguish between the $x$
and the $\xi$ components of the flow, so we will write
\begin{equation}
\label{xxi}
\phi^\vf_{a_1}(x,\xi)=(X^\vf(x,\xi),\Xi^\vf(x,\xi))\ .
\end{equation}

\begin{lemma}
	\label{homoa}
	$\forall \vf\in\R$ the function $X^\vf$ is quasi-homogeneous of degree
	1 and $\Xi^\vf$ is quasi-homogeneous of degree $\ell$.
\end{lemma}
\proof For $\lambda>0$, consider
$$
X_\lambda^\vf(x,\xi):=\frac{1}{\lambda}X^\vf(\lambda x,\lambda^\ell
\xi)\ ,\quad
\Xi_\lambda^\vf(x,\xi):=\frac{1}{\lambda^\ell}X^\vf(\lambda
x,\lambda^\ell \xi)\ . 
$$
One has
$$
X_\lambda^\vf(x,\xi)\big|_{\vf=0}=x\ ,\quad
\Xi_\lambda^\vf(x,\xi)\big|_{\vf=0}=\xi\ . 
$$
Furthermore, since $\frac{\partial a_1}{\partial \xi}$ is quasi
homogeneous of degree 1, one gets
\begin{align}
\label{qo.q}
\frac{d X_\lambda^\vf}{d\vf}(x,\xi)=\frac{1}{\lambda}\frac{d
	X^\vf}{d\vf}(\lambda x,\lambda^\ell\xi)=\frac{1}{\lambda }
\frac{\partial a_1}{\partial \xi}(X^\vf(\lambda x,\lambda^\ell
\xi),\Xi^\vf(\lambda x,\lambda^\ell\xi))
\\
= \frac{\partial a_1}{\partial \xi}\left(\frac{1}{\lambda} X^\vf(\lambda x,\lambda^\ell
\xi),\frac{1}{\lambda^\ell}\Xi^\vf(\lambda x,\lambda^\ell\xi)\right)=
\frac{\partial a_1}{\partial \xi}(X^\vf_\lambda(x,
\xi),\Xi_\lambda^\vf(x,\xi))
\ ,
\end{align}
which shows that $(X^\vf,\Xi^\vf)$ and
$(X_\lambda^\vf,\Xi_\lambda^\vf)$ satisfy the same Cauchy problem, so
they coincide. \qed

It follows that if $f$ is quasi-homogeneous of some degree, then
$f\circ\phi^\vf_{a_1}$ is quasi-homogeneous of the same degree. It is
also easy to see that if $f\in S^m_{HR}$ then the same is true for
$f\circ \phi^\vf_{a_1}$ (more precisely, according to our abuse of
notation, for $(f\circ \phi^\vf_{a_1})_\chi$). We are now going to show that
also the class $\Snoi^m$ is left invariant by the composition
with $\phi$ (up to a cutoff close to the origin).

\begin{lemma}
	\label{egor}
	If $f\in \Snoi^m$ then $f\circ \phi^{\vf}_{a_1} \in \Snoi^m $. 
\end{lemma}
\proof We are going to study $\frac{\partial^{|\alpha|}}{\partial
	x^\alpha}f\circ \phi^\vf_{a_1}$; the other derivatives can be studied in
the same way. We are going to prove that
\begin{equation}
\label{ricor}
\frac{\partial^{|\alpha|}}{\partial
	x^\alpha}f\circ
\phi^\vf_{a_1}=\sum_{\beta,\gamma}\Phi_{\gamma\beta}\frac{\partial^{|\gamma+\beta|}f}{\partial
	x^\gamma\partial\xi^\beta}\circ \phi^\vf_{a_1} \ 
\end{equation}
with $\Phi_{\gamma\beta}\in
S_{HR}^{|\beta|\delta_2+(|\gamma|-|\alpha|)\delta_1}$, and
$\left|\gamma\right|+\left|\beta\right|\leq \left|\alpha\right|$. From
this structure one immediately has
$$
\left|\frac{\partial^{|\alpha|}}{\partial
	x^\alpha}f\circ \phi^\vf_{a_1}\right|\sleq \tk_0^{m-|\alpha|\delta_1}\ ,
$$
and, {adding the estimates of the other derivative,} the thesis. We now prove \eqref{ricor}. To shorten the notation
denote  $f_{\gamma\beta}:=\frac{\partial^{|\gamma+\beta|}f}{\partial
	x^\gamma\partial\xi^\beta}$ and compute (to be determined)
\begin{align*}
\frac{\partial}{\partial x_1}\left(f_{\gamma\beta}\circ\phi^\vf_{{a_1}} \Phi_{\gamma\beta}\right) &=
\left(\frac{\partial X^\vf}{\partial x_1}\cdot \frac{\partial
	f_{\gamma\beta}}{\partial x} \circ\phi^\vf_{a_1}+\frac{\partial \Xi^\vf}{\partial x_1}\cdot \frac{\partial
	f_{\gamma\beta}}{\partial \xi}
\circ\phi^\vf_{a_1}\right)\Phi_{\gamma\beta}
\\
&+f_{\gamma\beta}\circ \phi_{a_1}^\vf \frac{\partial
	\Phi_{\gamma\beta}}{\partial x_1}\ .
\end{align*}
The second line has the wanted structure. We compute the structure of
the first line. Consider its first term (namely $\frac{\partial X^\vf}{\partial x_1}\cdot \frac{\partial
	f_{\gamma\beta}}{\partial x} \circ\phi^\vf_{a_1}   $) which is
actually the sum of two terms, corresponding to the different
components of the gradient and of  $X$. The first of these two terms has the
wanted structure with
$$
\Phi_{(\gamma+(1,0))\beta}:=\frac{\partial (X^\vf)_1}{\partial
	x_1}\Phi_{\gamma\beta} \in S^{|\beta|\delta_2+(|\gamma|-|\alpha|)\delta_1}_{HR}
$$
(since $\frac{\partial (X^\vf)_1}{\partial
	x_1}$ is homogeneous of degree zero), which has correct order since
$|\gamma+(1,0)|=|\gamma|+1$ and the new $\alpha$ has modulus
$|\alpha|+1$. Consider now the first of the terms containing $\Xi$. It
has the wanted structure with
$$
\Phi_{\gamma(\beta+(1,0))}:=\frac{\partial (\Xi^\vf)_1}{\partial
	x_1}\Phi_{\gamma\beta} \in
S^{|\beta|\delta_2+(|\gamma|-|\alpha|)\delta_1+\frac{\ell}{\ell+1}-\frac{1}{\ell+1}
}_{HR} 
$$
(since $\Xi^\vf$ is homogeneous of degree $\ell$). So the order is
$$
\left|\beta\right|\delta_2+(\left|\gamma\right|-(\left|\alpha\right|+1))\delta_1
+\delta_1+\frac{\ell-1}{\ell+1}=
(\left|\beta\right|+1)\delta_2+(\left|\gamma\right|-(\left|\alpha\right|+1))\delta_1
$$
as it should be. \qed

From now on we will need to consider the composition of the flows of
$a_1$ and $a_2$. So we denote
\begin{equation}
\label{vectora}
\vf:=(\vf_1,\vf_2) \ ,\quad
\phi_{a}^{\vf}:=\phi_{a_1}^{\vf_1}\circ\phi_{a_2}^{\vf_2}\ .
\end{equation}
and remark that {$\phi_{a}^\vf$} is $2\pi$ periodic in each one of the 
$\vf$'s.
\begin{remark}
	\label{egora}
	The result of Lemma \ref{egor} holds also for the
	joint flow $\phi^\vf_a$. 
\end{remark}

\subsection{Fourier expansion and summable symbols}\label{expans}

Following \cite{Bam96} we consider the Fourier
development defined by the flow $\phi_a^{\vf}$. 

\begin{definition}
	\label{fou.1}
	Let $f\in \Snoi^m$, then, {for $k\in\Z^2$,}
	\begin{equation}
	\label{fou.2}
	\hat
	f_k(x,\xi):=\frac{1}{4\pi^2}\int_{\T^2}f(\phi^\vf_a(x,\xi))e^{-\ii
		k\cdot \vf}d\vf
	\end{equation}
	is called the $k$-th Fourier coefficient of $f$.
\end{definition}

\begin{remark}
	\label{tuttaf}
	From standard theory of Fourier expansions one has, {for fixed
        $(x,\xi)\in\R^4$, }
	$$
	f(\phi_a^{\vf}(x,\xi))=\sum_{k\in\Z^2}\hat f_k(x,\xi)e^{\ii k\cdot
		\vf}\ ,
	$$
	so that, in particular
	\begin{equation}\label{four.somma}
	f(x,\xi)=\sum_{k\in\Z^2}\hat f_k(x,\xi)\ .
	\end{equation}
\end{remark}
\begin{remark}
  \label{inva}
{ Let $g$ be a function invariant under the flow of the actions, then, for any function $f$, one has that the $k$-th
Fourier coefficient of $gf$, namely $\widehat{(gf)}_k$, is $g\hat f_k$.  }
\end{remark}

The $0-$th Fourier coefficient of a symbol $f \in \Snorma^{m}$ is
essentially a function of the actions only. To give a precise
statement, we proceed 
as in \cite{chaspe}; introduce an open cone
$\cV$, s.t. $\overline{\cV}\setminus\left\{0\right\} $ is contained in the
interior of $\Pi$ and $\overline{\cC}\setminus\left\{0\right\} $ is contained in the
interior of $\cV$, then we have the following Lemma.

\begin{lemma}
  \label{lem.a.dependence}
  {Let $f\in\Snoi^m$, with the property that
    $f\circ\phi^\vf_a=f$, $\forall \vf\in\T^2$,
then there exists a symbol $f_c\in \sm {m}$,
with $\varsigma=1-(M-\delta)$, s.t.}
\begin{equation}\label{actions.only}
f(x,\xi)=f_c(a(x,\xi))\ ,\quad {\forall (x,\xi)\ :
  \ a(x,\xi)\in\cV\cap B_1^c}\,,
\end{equation}
{where, given $r>0$, $B_r^c$ is the complement of the ball of radius $r$ in $\R^2$.}
\end{lemma}
The proof of this Lemma is postponed to Appendix \ref{tutto.sbagliato}. 

{From now on, if $f:\R^4 \rightarrow \R$ is invariant with respect to
  the flow of the actions, 
 we will simply say that it is a function of the actions only, {and we
   will denote by $f_c\in\sm m$ the symbol such that \eqref{actions.only} holds}.
}

\begin{definition}\label{def.norma}
For $m\in\R$ and $N\in\N$, the set of the symbols $f \in \Snoi^m$
s.t. {$\forall \alpha,\beta\in\N^2$,} 
	\begin{gather}
	\label{norma.tutto}
	\norfou{f}{m}{\alpha, \beta, N}:= \sum_{k \in \Z^2} \langle k
        \rangle^{N} \norfou{\hat{f}_k}{m}{\alpha, \beta}<\infty\,, 
        \end{gather}
will be denoted by $\Snorma^m$.  The seminorms $
\norfou{\cdot}{m}{\alpha, \beta}$ are defined by \eqref{norma.k}.
 \\
We will say that an operator $F$ is a
pseudo-differential operator of class $\OPSFnorma^{m}$ if there exists symbol $f \in \Snorma^{m}$ s.t. $F=\Op (f)$.
\end{definition}

The following result relates {the symbols of class $\Shr^m$ with those
of class $\Snorma^m$:}
\begin{lemma}
	\label{foou.11}
	Let $f \in \shr m$, then for any $\alpha, \beta \in \N^2$ and $N\in \N$ there exists $C>0$, independent of $f,$ such that
		\begin{equation}\label{easy}
		\sup_{x, \xi} \sup_{|\alpha^\prime| \leq |\alpha|, |\beta^\prime| \leq |\beta|} \left| \partial_{x}^\alpha \partial_\xi^\beta \hat{f}_k(x, \csi)\right| (\tk_0(x, \csi))^{-(m-\frac{|\alpha|}{\ell + 1}- \frac{|\beta|\ell}{\ell + 1})} \leq C \langle k \rangle^{-N}\,.
		\end{equation}
		In particular, $f \in \Snorma^m$ for all $N$.
\end{lemma}

\proof  In order to prove that \eqref{easy} holds, we proceed as usual,
by integrating by parts in \eqref{fou.2}. To fix ideas consider the
case $k_1>0$, $k_2>0$. Denote \emph{just for this proof},
$\partial:=\partial_{\vf_1}+\partial_{\vf_2} $, then exploiting
$\partial e^{-\ii k\cdot\vf}= -\ii |k| e^{-\ii k\cdot\vf}$, we have
$\forall N\geq 1$
\begin{align*}
\hat f_k=\frac{1}{4\pi^2}\int_{\T^2}(f\circ\phi^\vf_a)\frac{1}{(-\ii
	|k|)^N}\partial^N e^{\ii k\cdot\vf}d\vf\ 
\\
=\frac{1}{(\ii
	|k|)^N}\frac{1}{4\pi^2}  \int_{\T^2}\partial^N(f\circ\phi^\vf_a)
e^{\ii k\cdot\vf}d\vf\ . 
\end{align*}
In order to estimate $\partial^N(f\circ\phi^\vf_a)$, we start by
considering
$$
\partial_{\vf_i} (f\circ\phi^\vf_a)=\left\{f;a_i\right\}\circ
\phi^{\vf}_{a} =(ad_{a_i}f)\circ\phi^{\vf}_{a}\ ,
$$
where
$$
ad_{a_i}f:=\left\{f;a_i\right\}\,.
$$
From this one has
$$
\partial^N(f\circ\phi^\vf_a)=\left[\left(ad_{a_1}+ad_{a_2}\right)^N
f\right]\circ\phi^{\vf}_{a} 
\ .
$$
Now, if $f \in \Shr^m$ one has
$$
ad_{a_i}f=\frac{\partial a_i}{\partial \xi}\cdot \frac{\partial
	f}{\partial x}- \frac{\partial a_i}{\partial x}\cdot \frac{\partial
	f}{\partial \xi} \in \Shr^{m}
$$
due to Lemma \ref{lemma composizione}, thus  $\partial^N(f\circ\phi^\vf_a)\in \Shr^{m}$ for any $N$.
By Lemma \ref{egor}, this implies
$$
\partial^N(f\circ\phi^\vf_a)\in\Shr^{m}\ ,
$$
from which \eqref{easy} immediately follows.
\qed

\begin{remark}
	We point out that $f \in \Snoi^m$ does not imply $f \in \Snorma^{m}$; in particular, arguing as in the proof of Lemma \ref{foou.11}, from $f \in \Snoi^m$ one can only deduce that 
	$$
	\norfou{f}{m + (M-\delta)N}{\alpha, \beta, N} < +\infty \quad \forall \alpha, \beta \in \N^2, \quad \forall N \in \N\,.
	$$
	This is the main reason why we work in the stronger class $\Snorma^m$.
\end{remark}
\begin{lemma}\label{lem.prod.n}
	Let $f \in \Snorma^{m}$ and $g \in \Snorma^{m^\prime}.$ Then
        $fg \in \Snorma^{m+m^\prime}$, and  $\forall N \in \N$ there
        exists $C_N>0$ such that 
	$\forall \alpha, \beta \in \N^2$ 
	\begin{equation}\label{fou.prodotto}
	\norfou{f g}{m + m^\prime}{\alpha, \beta, N} \leq C_N \norfou{f}{m}{\alpha, \beta, N} \norfou{g}{m^\prime}{\alpha, \beta, N}\,.
	\end{equation}
\end{lemma}

\begin{proof}
	First of all we observe that, 
	\begin{gather*}
(fg)\circ \phi^\vphi_a=	\left(f \circ \phi^\vphi_a\right) \left(g \circ \phi^\vphi_a\right)
	= \sum_{k, k^\prime \in \Z^2} \hat{f}_{k} \hat{g}_{k^\prime} e^{\im (k+k') \cdot \vphi}\,,
	\end{gather*}
        thus we have
        \begin{align}
              \label{stim.prod}
\norfou{fg}{m+m'}{0,0,N}\leq \sum_{k,k'}\langle k+k'\rangle^N
\left|\hat f_k\right|\left|\hat g_k\right| \tk_0^{-(m+m')} 
\\
\label{stim.prod.1}
\leq C_N \sum_{k,k'}\langle k\rangle^N\langle k'\rangle^N
\left|\hat f_k\right|\left|\hat g_k\right| \tk_0^{-(m+m')} \leq C_N \norfou{f}{m}{0,0,N}\norfou{g}{m'}{0,0,N}\,.
        \end{align}
Working in the same way with the derivatives and exploiting Remark
\ref{prod 2}, one gets the result.
\end{proof}
The following results {can be proved by simple variations
on the standard arguments}


	\begin{lemma}\label{lemma composizione}
		Let $m, m' \in \R$,  $F = \op f \in \OPSFnorma^{m}$, $G =
		\op g \in \OPSFnorma^{m'} $. Then  $FG \in \OPSFnorma^{m +
			m'}$. Denote by $f\sharp g$ its symbol, then
                it admits the same asymptotic expansion as in
                \eqref{sigma1 2}, but now
				\begin{align}
                                    \label{sigma2}
		(f\sharp g)_j \in \Snorma^{m + m' -
                                      (\delta_1+\delta_2) j} {\quad \forall j \in \N}\ .
                   \end{align}
	\end{lemma}
        \begin{corollary}
          \label{moy}
One has
		$$
		\moyal fg=  \poisson fg+\Snorma^{m+m'-3(\delta_1+\delta_2)}\ .
		$$
        \end{corollary}

        \begin{definition}
          \label{moyals}
          In what follows we will denote
$$
\adm ga:=\moyal ag \ ,
$$
which is well defined in anyone of the classes of symbols we are using.
        \end{definition}

Given a self-adjoint pseudo-differential operator $G\in\OPSFnorma^ \eta$, we consider
the unitary group generated by $-\im G$, which is denoted, as
usual, by $\es^{-\im\tau G}$, $\tau\in\R$. \\
The following version of Egorov Theorem holds:
\begin{lemma}\label{teo Egorov}
	Fix $\eta \in \R$, and let $g\in \Snorma^{\eta}$ be a real valued symbol,
	denote $G = \op{g}$, then $\forall \tau \in [-1, 1]$
	\begin{itemize}
		\item[(1)] If $\eta \leq \delta + \dfrac{2}{\ell + 1}$, then $e^{\ii \tau G}
		\in {\cal B}\left({\cH}^{s}; {\cH}^{s}\right)
		\quad \forall\ s \geq 0$
		\item[(2)] {Assume  $\eta < \delta_1 + \delta_2$, and
                  let  $f\in
		  \Snorma^m$,  $F =\op f$, then
		$$
		e^{\ii \tau G} F e^{- \ii \tau G}=:F'\in \OPSFnorma^m\,.
		$$
Furthermore, denoting by $f'$ its symbol, for any $n \in \N$ one has}
		\begin{equation}
		\label{asy.moy}
		f' = \sum_{0\leq j \leq  n} \frac{ \tau^j(\adm g)^jf}{j ! } + \Snorma^{m + (n+1)(\eta-(\delta_1 + \delta_2))}\, .
		\end{equation}
		In particular, one has
		\begin{equation}\label{espansione egorov astratto}
		f' = f  + \moyal{f}{g} + \Snorma^{m + 2(\eta - \delta_1-\delta_2)}\,. 
		\end{equation}
	\end{itemize}
\end{lemma}
\begin{definition}
	\label{ilprimo}
	As a general notation, given two
	symbols $f\in\Snorma^m$ and $g\in\Snorma^\eta$ with
	$\eta<\delta_1+\delta_2$, we will denote by $f'\in\Snorma^m$ (namely with a prime)
	the symbol such that 
	$$
	e^{\im G}\Op(f)e^{-\im
		G}=\Op(f')\,.
	$$
\end{definition}

\section{The normal form lemma}\label{nfl}

From now on we abandon our abuse of notation related to the cutoff at
the origin and a symbol will always be a $C^{\infty} $ function fulfilling the
required estimates everywhere.\\
We give the following definition:
\begin{definition}\label{def.nf}
	We say that a symbol $z \in \Snorma^m$ is in resonant normal form if $\forall k \in \Z^2 \backslash\{0\}$ its Fourier coefficients satisfy the following:
	\begin{equation}
	\textrm{supp}\,\hat{z}_k \subseteq \cR_k := \left\lbrace (x,\xi)\ |\ |\omega(a)\cdot k| \leq \|a\|^\delta \|k\|\,, \quad \|k\| \leq \|a\|^\ep\right\rbrace\,,
	\end{equation}
where $supp$ denotes the support of the function in argument.
\end{definition}
This section is devoted to the proof of the following result:
\begin{lemma}[Normal form]\label{lemma.in.forma}
	Let $m = \frac{2 \ell}{\ell + 1}- \fke,$ $\fke >0.$ 
	Furthermore, 
	suppose that $\ep, \delta$ satisfy \eqref{delta} and that
	\begin{equation}\label{legami.tra.parametri.1}
	M -\delta < \min\left\{\frac{\fke}{3},\,\frac{2}{7}\right\}\,,
	\end{equation}
	and define, {as in Eq. \eqref{def.varsigma},}
	\begin{equation}\label{quanto.guadagno}
	\rho =  \min\{\fke- 3(M-\delta),\,2-7(M-\delta)\}\,.
	\end{equation}
	
	Fix $\tN$ arbitrarily large, take $N$ s.t. {$N\ep\geq\tN \rho $,}
	then there exists a sequence of self-adjoint
	pseudo-differential operators $\{G_n\}_{n=1}^{\tN-1}$, $G_n \in
	\OPSFnorma^{m-n\rho-\delta}$, such that the operator
	$$
	\cU_{n} := e^{\im G_1} \circ \cdots \circ e^{\im G_n}
	$$
	conjugates $H$ to $H_n = \Op(h_n + w_n)$, where $w_n \in {\cal SF}_{0}^{m-\tN \rho}$ is a real symbol and
	$h_n = h_0 + z^{(n)}  + v_n$
	has the following properties:
	\begin{enumerate}
		\item \label{piccolo} $v_n \in \Snorma^{m -n\rho}$ is
                  a real symbol  
		\item \label{normale} $z^{(n)} = \langle z^{(n)} \rangle + z_n^{(res)}$ is in resonant normal form,
		and there exists a sequence $\{z_j\}_{j \in \N}$ of
                real valued smooth functions ${z_j \in \Snorma^{m-
                    {(j-1)}\rho},}$ which are functions of the actions
                only, such that
		\begin{equation}
		\langle z^{(n)} \rangle = \sum_{j=1}^n z_j\,. 
		\end{equation}
	\end{enumerate}
\end{lemma}

The proof of Lemma \ref{lemma.in.forma} is obtained following the same approach of \cite{BLMnr, PS10}, which we are now going to adapt.\\
Consider again the cutoff function $\chi$ fixed above.  With its help
we define, for $k\in\Z^2\setminus\left\{0\right\}$

\begin{equation}\label{cut-off-piccoli-divisori}
\begin{gathered}
\chi_k( a) := \chi\Big( \frac{\omega(a) \cdot
	k}{\| a\|^{ \delta}\|k\|} \Big)\ ,
\\
d_k( a) := \frac{1}{\im \omega(a) \cdot k}(1-\chi_k(a))\ ,
\\
\tilde\chi_k( a) := \chi \left(\frac{\|k\|}{\|a\|^{\epsilon}}\right)\ . 
\end{gathered}
\end{equation}
Of course they will be considered as functions of
$(x,\xi)$ (by the substitution $a=a(x,\xi)$).
Furthermore, 
given $f \in \Snorma^m$, we define 
\begin{equation}\label{def a r nr R}
\begin{aligned}
&  \langle f \rangle:=(1-\chi(\left\|a\right\|)) \hat f_0\,, \\
&	f^{(res)}:= \sum_{k \in Z^2 \setminus \{ 0 \}}(1-\chi(\left\|a\right\|))
\chi_k \tilde\chi_k\hat{f}_k\ , \\
& f^{(nr)}:= \sum_{k \in Z^2}(1-\chi(\left\|a\right\|))(1-\chi_k) \tilde\chi_k
\hat{f}_k\ ,\\
& f^{(S)} := \sum_{k \in Z^2 \setminus \{ 0
  \}}(1-\chi(\left\|a\right\|))(1-\tilde\chi_k)\hat{f}_k+\chi(\left\|a\right\|)
f\ ,
\end{aligned}
\end{equation}
so that one has 
\begin{equation}\label{splitting simbolo a}
f = \langle f \rangle + f^{(nr)} + f^{(res)} + f^{(S)}\,.
\end{equation}

In order to show that each term is a symbol, we need a few
preliminaries.

\begin{lemma}\label{stima g chi f psi}
One has ${(1-\chi(\|a\|))\chi_k\in\Snoi^0}$, furthermore {for any
$\alpha,\beta$ there exists $C$ s.t.} $\forall k\in\Z^2\setminus\left\{0\right\}$ one has
$$
\norfou{{(1-\chi(\|a\|))\chi_k}}0{\alpha,\beta}\leq C\ .
$$
\end{lemma}
\proof
{First of all, observe that since $(1-\chi(\|a\|))$ is a symbol {of class
  $\Shr^0$} it
is enough to study $\chi_k$ \emph{in the support of $(1-\chi(\|a\|))$.}} Denote
$$
t_k(a):= \frac{\omega(a) \cdot
	k}{\|k\| \|a\|^{ \delta}},
$$
which is a homogeneous function of degree $M-\delta>0$; {then one has $\chi_k = \chi \circ t_k$}, and by the Fa\`{a} di Bruno formula
\begin{equation}
  \label{derchi}
\left|\partial_x^\alpha(\chi\circ t_k)\right|\simeq
\sum_{j=1}^{|\alpha|}
\sum_{\gamma_1+...+\gamma_j=\alpha}\left|\chi^{(j)}\circ
t_k\right|\prod_{i=i}^{j }\left|\partial_x^{\gamma_i}t_k\right|\ .
\end{equation}
Now, by quasi homogeneity, one has 
$$
\left|\partial_x^{\gamma_i}t_k\right|\sleq 
\|a\|^{M-\delta-\frac{|\gamma_i|}{\ell+1}}\sleq 
\|a\|^{|\gamma_i|(M-\delta-\frac{1}{\ell+1})}\sleq 
\|a\|^{-|\gamma_i|\delta_1}\ ,
$$
and therefore 
\begin{equation}
  \label{snnon}
\prod_{i=i}^{j }\left|\partial_x^{\gamma_i}t_k\right|\sleq
\|a\|^{-|\alpha|\delta_1}\ .
\end{equation}
{Remark now, that, in the support of $1 -\chi(\|a\|)$, by Remark \ref{equia} one has ${\langle a \rangle} \simeq \tk_0$, thus one has $
\prod_{i=i}^{j }\left|\partial_x^{\gamma_i}t_k\right|\sleq
\tk_0^{-|\alpha|\delta_1}
$, which in turn implies
$$
\left|\partial_x^\alpha(\chi\circ t_k)\right|\sleq
\tk_0^{-|\alpha|\delta_1}\qquad \forall (x,\xi)\ ,\quad \forall k\not=0 \ . 
$$}
Similar estimates hold for the $\xi$ derivatives and for the mixed
derivatives, and this implies the thesis.  \qed

\begin{lemma}\label{stima tilde}
  One has $\tilde\chi_k\in\shr0$, furthermore all its seminorms are
  bounded uniformly with respect to $k$.
\end{lemma}
\proof We proceed as in the proof of the preceding Lemma, except that
we redefine
$$
t_k:=\frac{\|k\|}{\|a\|^{\epsilon}}\ .
$$
Thus one gets again formula \eqref{derchi}. In this case one has
$$
\left|\partial_x^{\gamma_i}t_k\right|=\|k\|\left|\partial_x^{\gamma_i}\frac{1}{\|a\|^{\epsilon}}\right|\sleq
\|k\|\|a\|^{-(\epsilon+\frac{|\gamma_i|}{\ell+1})}  
$$
and thus,
$$
\prod_{i=1}^{j}\left|\partial_x^{\gamma_i}t_k\right|\sleq 
\left(\frac{\|k\|}{\|a\|^{\epsilon}}\right)^j\frac{1}{\|a\|^{\frac{|\alpha|}{\ell+1}}}\ ,   $$
but, on the support of $\tilde \chi\circ t_k$, one has $
\dfrac{\|k\|}{\|a\|^{\epsilon}}<1 $, which also implies
${\|a\|>(\|k\|)^{1/\epsilon}\geq 1,}$ thus we also have, on the whole of $\R^4$,
$$
\left|\partial^\alpha_x(\chi\circ t_k)\right|\sleq
\frac{1}{\tk_0^{\frac{|\alpha|}{\ell+1}}}\ .
$$
Similar estimates hold for the $\xi$ derivatives, for the mixed
derivatives and this implies the thesis.  \qed

\begin{remark}
	\label{classe}
	The results of this Lemma and of Lemma \ref{dk} below are the main reason for the introduction of
	the class of symbols of Definition \ref{symbol.a1}
\end{remark}
\begin{lemma}\label{dk}
One has $(1-\chi(\|a\|))d_k\in\Snoi^{-\delta}$, furthermore all its seminorms are
  bounded uniformly with respect to $k$.
\end{lemma}

\proof We study the $x$ derivatives of $d_k$, all the others can be
estimated exactly in the same way. {Arguing as in the proof of Lemma \ref{stima g chi f psi}, again we restrict the study of $d_k$ to the support of $(1-\chi(\|a\|))$.}
From Leibniz formula we have
$$
|\partial_x^\alpha d_k|\simeq \sum_{\gamma\leq
	\alpha}\left|\partial^\gamma_x\frac{1}{\omega\cdot
	k}\right| |\partial^{\alpha-\gamma}_x(1-\chi_k)|\ .
$$
Now, one has by Fa\`{a} di Bruno formula
$$
\left|\partial^\gamma_x \frac{1}{\omega\cdot
	k}
\right| \simeq
\sum_{j=1}^{|\gamma|}\sum_{\nu_1+...+\nu_j=\gamma}\frac{1}{|\omega\cdot
	k|^{j+1}}\prod_{i=1}^{j}\left|\partial^{\nu_i}_x\omega\cdot k\right|\ . 
$$
Now, since we work in the support of $1-\chi_k$ which contains also
the support of $\partial_x^{\alpha-\gamma}\chi_k$, we have, in this
domain intersected with $\|a\|>\frac{1}{2}$,
$$
\frac{1}{|\omega\cdot k|^{j+1}}\sleq \|a\|^{-\delta (j+1)} \|k\|^{-(j+1)}\sleq \tk_0^{-\delta (j+1)} \|k\|^{-(j+1)}\ ,\quad
|\partial^{\nu_i}_x\omega\cdot k| \sleq
\tk_0^{M-\frac{|\nu_i|}{\ell+1}}\|k\| \ ,
$$
and therefore
\begin{equation}
\label{stipd}
\left|\partial^\gamma_x \frac{1}{\omega\cdot
	k}
\right|
\sleq \tk_0^{-\delta+|\gamma|(M-\delta-\frac{1}{\ell+1})}\sleq
\tk_0^{-\delta-|\gamma|\delta_1}\ .
\end{equation}
From this, using Lemma \ref{stima g chi f psi}, one finally gets
$$
|\partial_x^\alpha d_k|\sleq
\tk_0^{-|\alpha|\delta_1-\delta}\ .
$$
Performing the analogous estimates for the $\xi$ derivatives one gets
the result. \qed

We are now ready to prove the following:

\begin{lemma}\label{lemma.f.a.pezzi}
	Let $f \in \Snorma^{m}$;
	then $\langle f \rangle$, $f^{(res)}$ and $f^{(nr)}$ are in
        $\Snorma^{m}$.
\end{lemma}
\proof First we remark that, since the cutoffs are functions of $a$
only, Remark \ref{inva} allows to compute the Fourier coefficients of
the different parts of $f$. 

Consider $f^{(res)}$. Since the seminorms of the cutoffs are bounded
uniformly with respect to $k$, one has
\begin{align}
  \label{stimafres}
\left\|f^{(res)}\right\|^{(m)}_{\alpha,\beta,N}=\sum_{k}
\left\|\chi_k\tilde \chi_k (1-\chi(a))\hat
f_k\right\|^{(m)}_{\alpha,\beta} \langle k\rangle^N\sleq \sum_{k} \left\|\hat
f_k\right\|^{(m)}_{\alpha,\beta} \langle k\rangle^N\sleq
\left\|f\right\|^{(m)}_{\alpha,\beta,N}\ ,
\end{align}
so that $f^{(res)}\in \Snorma^m$.

The other parts of $f$ can be estimated exactly in the same way.
\qed

        \begin{lemma}
          \label{smoothing}
 Let {$N\ep\geq\tN \rho$,} then one has $f^{(S)} \in {\mathcal{SF}^{m-\tN\rho }_0}$.
        \end{lemma}
\proof First we remark that, since $\chi(\|a\|)f$ has compact support, the
thesis trivially holds for this part of $f^{(S)}$. For the other part, it is enough
to observe that on $\textrm{supp}\,(1-\tilde{\chi}_k(x, \csi))$ we have
	\begin{equation}\label{k.large}
	\|k\| \geq \left(a(x, \csi)\right)^{\ep} \gtrsim \left(\tk_0(x, \xi)\right)^{\ep}\,
\end{equation}
       by Remark \ref{equia}, so that one has
        \begin{align*}
\left\|(1-\chi(\|a\|))(1-\tilde \chi_k)\hat f_k\right\|^{(m-\epsilon
  N)}_{\alpha,\beta} &\sleq
\sup_{|\alpha'|\leq|\alpha|,|\beta'|\leq|\beta| }
\sup_{x,\xi}\frac{\left|\partial_x^{\alpha'}\partial_\xi^{\beta'}
  (1-\chi(\|a\|))(1-\tilde \chi_k)\hat f_k \right|}{\tk_0^{m-\epsilon
    N}}
\\
=\sup_{|\alpha'|\leq|\alpha|,|\beta'|\leq|\beta| }
&\sup_{x,\xi}\frac{\left|\partial_x^{\alpha'}\partial_\xi^{\beta'}
  (1-\chi(\|a\|))(1-\tilde \chi_k)\hat f_k
  \right|}{\tk_0^{m}}\frac{\tk_0^{\epsilon N}}{\langle
  k\rangle^N}\langle k\rangle^N
\\
\sleq 
\sup_{|\alpha'|\leq|\alpha|,|\beta'|\leq|\beta| }
&\sup_{x,\xi}\frac{\left|\partial_x^{\alpha'}\partial_\xi^{\beta'}
  (1-\chi(\|a\|))(1-\tilde \chi_k)\hat f_k
  \right|}{\tk_0^{m}}
\langle k\rangle^N\\
&\sleq \left\|\hat f_k\right\|^{(m)}_{\alpha,\beta}\langle k\rangle^N
    \end{align*}
from which, summing over $k$ one gets the thesis.  
\qed

\begin{lemma}\label{lem.hom}
	Let $f \in \Snorma^{m}$ a real valued symbol. Then the equation
	\begin{equation}\label{solve.me}
	\{ h_0; g\} + f^{(nr)} = 0
	\end{equation}
	has a real valued solution $g \in\Snorma^{m-\delta},$ defined by
	\begin{equation}\label{sol.hom}
	g(x, \xi) :=\sum_{k \in \Z^2 \backslash \{0\}}  d_k(x, \xi) \tilde{\chi}_k(x, \xi)(1-\chi(\|a\|)) \hat{f}_k(x, \csi)\,.
	\end{equation}
\end{lemma}
\begin{proof}
{First we verify that $g$ solves \eqref{solve.me}}. One has
	$$
	\poisson{h_0}{g} = \left.\frac{d}{d t}\right|_{t=0} g \circ \phi^{\omega t}_a = \left.\frac{d}{d t}\right|_{t=0} \sum_{k \in \Z^2} {\hat{g}_k} e^{\im k \cdot \omega t} = \sum_{k \in \Z^2} \im (\omega \cdot k )\hat{g}_k\,,
	$$
	thus, recalling the definition of $f^{(nr)}$, Equation \eqref{solve.me} reads
	$$
	\im (\omega \cdot k) \hat{g}_k = \tilde{\chi}_k (1-\chi_k)
        (1-\chi(\|a\|))\hat{f}_k \,,
	$$
	which immediately implies that $g$ defined as in
        \eqref{sol.hom} solves \eqref{solve.me}. In order to prove
        that $g \in \Snorma^{m -\delta},$ one argues as in the proof
        of Lemma \ref{lemma.f.a.pezzi}, namely
since the seminorms of the cutoffs and of $d_k$ are bounded
uniformly with respect to $k$, one has
\begin{align}
  \label{stimafres1}
\left\|g\right\|^{(m-\delta)}_{\alpha,\beta,N}=\sum_{k}
\left\|d_k\tilde \chi_k (1-\chi(\|a\|))\hat
f_k\right\|^{(m-\delta)}_{\alpha,\beta} \langle k\rangle^N
\\
\sleq \sum_{k} \left\|
d_k\right\|^{(-\delta)}_{\alpha,\beta} \left\|
\hat f_k\right\|^{(m)}_{\alpha,\beta} \langle k\rangle^N\sleq
\left\|f\right\|^{(m)}_{\alpha,\beta,N}\ . 
\end{align}
\end{proof}

We are finally able to prove Lemma \ref{lemma.in.forma}:

\begin{proof}[Proof of the normal form Lemma \ref{lemma.in.forma}]
	The proof is obtained working inductively. For $n=0$,
        the thesis holds true with $z_n = 0$ and $v_n = v$; recall
        indeed that, by Lemma \ref{foou.11}, $v \in \Shr^{m} \subset
        \Snorma^{m}$.
We now construct a pseudo-differential operator $G_{n+1}$ with symbol
$g_{n+1}\in \Snorma^{m-n\rho-\delta}$ such that $H_{n+1}=e^{\im G_{n+1}}H_ne^{-\im
	G_{n+1}}$. Using the notation \ref{ilprimo}, we have $H_{n+1}=\Op(h'_n + w^\prime_n)$,
where $w^\prime_n \in {\cal SF}_0^{m-\tN \rho}$ by Lemma \ref{teo Egorov}, and $h^\prime_n$ is given by
\begin{align}
  \label{hprimo}
  h'_n&=h_0+z^{(n)}+v_n+\left\{h_0;g_{n+1}\right\}_M
  \\
  &+h'_0-h_0- \left\{h_0;g_{n+1}\right\}_M+z^{(n)\prime}-z^{(n)}+v'_n-v_n
  \\
  \label{nor.for.1}
  &= h_0+z^{(n)}+v_n^{(res)}+\langle v_n\rangle+v_n^{(S)}
  \\
\label{nor.for.2}
  & +v_n^{(nr)}+\left\{h_0;g_{n+1}\right\}
    \\
\label{nor.for.3}
  & + \left\{h_0;g_{n+1}\right\}_M-\left\{h_0;g_{n+1}\right\}
\\
\label{nor.for.4}
  &+h'_0-h_0- \left\{h_0;g_{n+1}\right\}_M
\\
\label{nor.for.5}
  &+z^{(n)\prime}-z^{(n)}
\\
\label{nor.for.6}
  &+v^{\prime}_n-v_n\ .
\end{align}
We use Lemma \ref{lem.hom} to construct $g_{n+1}$ in such a way that
\eqref{nor.for.2} vanishes. Then we define
\begin{align}
  \label{nouvofor}
{z_{n+1}:=\langle v_n\rangle\ ,}\quad  z^{(n+1)}:=z^{(n)}+v_n^{(res)}+z_{n+1}\ ,
  \\
  v_{n+1}:=\eqref{nor.for.3}+\eqref{nor.for.4}+\eqref{nor.for.5}
  +\eqref{nor.for.6}  
\end{align}
and 
\begin{equation}
w_{n+1} := w'_n + v_n^{(S)}\,.
\end{equation}
We now study the classes of the different lines. We
just compute the order of each term as a symbol in $\Snorma$. The order
of \eqref{nor.for.3} is
$$
\frac{2\ell}{\ell+1}+m-n\rho-\delta-3(\delta_1+\delta_2)=m-n\rho-\rho_1\ ,
$$
with
$$\rho_1=3(\delta_1+\delta_2)+\delta-\frac{2\ell}{\ell+1}=7(\delta-M)+2\geq
\rho\ .
$$
To estimate the order of \eqref{nor.for.4} we remark that, according
to \eqref{asy.moy} it is the same as the order of 
\begin{align}
  \label{nor.for.4.1}
\left\{\left\{h_0;g_{n+1}\right\}_M;g_{n+1}\right\}_M= 
\left\{\left\{h_0;g_{n+1}\right\};g_{n+1}\right\}_M
\\
\label{nor.for.4.2}
+\left\{\left\{h_0;g_{n+1}\right\}_M-\left\{h_0;g_{n+1}\right\};g_{n+1}\right\}_M\ . 
\end{align}
Now, exploiting the definition of $g_{n+1}$, the r.h.s. of
\eqref{nor.for.4.1} is equal to $-\left\{v_n^{(nr)};g_{n+1}\right\}_M
$ whose order is
$$
m-n\rho+m-n\rho-\delta-\delta_1-\delta_2= m-n\rho-\rho_2 
$$
with
$$
\rho_2=\delta+\delta_1+\delta_2-m\geq\rho\ .
$$
Concerning \eqref{nor.for.4.2}, its order is
$$
m-n\rho-\rho_1+m-n\rho-\delta-\delta_1-\delta_2<m-n\rho-\rho_1\ .
$$
Finally the order of \eqref{nor.for.5} is $m-n\rho-\rho_2$ and the
order of \eqref{nor.for.6} is the same as the order of
\eqref{nor.for.4.1}. This concludes the proof. 
\end{proof}


\section{Spectral result}\label{sec.qmodi}

In this section we prove the spectral asymptotic \eqref{asym} claimed
in Theorem \ref{maint}.

\begin{definition}
  \label{defomega}
	Let $\varsigma, \delta,\epsilon$ fulfilling \eqref{def.d.ep}
        and \eqref{def.varsigma}, $\tC$ as in \eqref{reticolo.cono}, define
	\begin{equation}\label{def.om.2}
	{\tilde\Omega} := \left\lbrace \ta \in \tC \ |\ |\omega(\ta)\cdot k| \geq 2 \|k\| \|\ta\|^{\delta} \quad \forall\ k\in \Z^2 \textrm{ s.t. } 0<\|k\|< 2\|\ta\|^{\ep} \right\rbrace\,.
	\end{equation}
\end{definition}

The main result of this section is the following 

\begin{theorem}\label{prop.victor.hugo}
	{Given ${\tt R}>0$, let
	\begin{equation}\label{def.om}
	\Omega  := \tilde{\Omega} \cap B_{\tt R}^c\,,
	\end{equation}
	with $\tilde{\Omega}$ as in \eqref{def.om.2}}. There exist a sequence of symbols $\{\tilde{z}_j\}_{j\geq 1}$, $\tilde{z}_j \in \sm{m-\rho j}$, such that {if ${\tt R}>0$ is big enough,}
	for any $\ta \in \Omega$ there exists an eigenvalue $\lambda_{\ta}$ of \eqref{ope1} which admits the asymptotic expansion
	\begin{equation}\label{asym.ripetuta}
	\lambda_\ta \sim \tilde{h}_0(\ta) + \sum_{j \geq 0} \tilde z_j(\ta)\,,
	\end{equation}
	where $\displaystyle{\tilde{h}_0(\ta) \sim \sum_{j \geq0}
          \tilde h_{0, j}(\ta)}$, is the function
        \eqref{cha.2.1}.  
\end{theorem}

The proof is based on a quasi-mode argument; in particular, our aim is
to prove that the joint eigenfunctions of $A_1$ and $A_2$ defined as
in \eqref{eige} are quasi-modes for the normal form operator $H_{\tN}$. Since $\tN$ is
arbitrary, the result follows. The first property we exploit is the following:
\begin{remark}
  \label{joint}
By the ellipticity of $ {\textrm{Id} + }A_1^2+A_2^2$, the joint
eigenfunctions $\psi_{\ta}$ of $A_1$ and $A_2$ (defined by
\eqref{eige}) satisfy
$$
\left\|\psi_{\ta}\right\|_{s}\lesssim_s {\langle a \rangle}^{s}\ ,\quad \forall s\in\R
$$
and therefore, if $R$ is a smoothing operator,
$$
\left\|R\psi_{\ta}\right\|_{L^2}\lesssim_n \frac{1}{{\langle a \rangle}^{n}}\ ,\quad \forall n\in
\N\ .
$$
\end{remark}


The second key property that we exploit for the proof concerns symbols
which are functions of the actions only.

First of all, we give the
following Lemma which is a variant of Theorem 1 of \cite{Cha83}:
\begin{lemma}[Theorem 1 of \cite{Cha83}]\label{cha.lem.1}
	Given $m \in \R$ and $0< \varsigma \leq 1$, let $f \in \sm{m}$, there exists a sequence of symbols
        $\{\check f_j\}_{j \in \N}$ with $\check f_j \in \sm{m-\varsigma j}$
    and $\check f_0 = f$ such that $\forall \tN \in \N$ 
	\begin{equation}
	f(A) = \sum_{0 \leq j < \tN} \Op(\check f_j\circ a) + R_\tN\,,
	\end{equation}
	where $R_\tN$ is a smoothing operator of order {$\varsigma \tN - m$}
        and $f(A)$ is spectrally defined.

        Moreover,
        \begin{itemize}
        \item[(1)]        $\textrm{supp\,}(\check f_j) \subseteq \textrm{supp\,}(f)$
          for all $j$
        \item[(2)] $\forall s$, $\exists \alpha,\beta$ and $C_\tN$, independent of $f$ s.t.
          $$\| R_\tN\|_{s,s - (m-\varsigma \tN)}\leq
          C_\tN\left\| f\right\|^{(m)}_{\alpha,\beta}\,.
          $$
                \end{itemize}
\end{lemma}
The proof is easily obtained remarking that in our context formula
(3.10) of \cite {Cha83} holds with $a_{fj}\in\sm{m-j\varsigma}$, which also allows to adapt the
estimate of the remainder done in Sect. 4 of that paper.

As a consequence, one has the following:
\begin{lemma}\label{cha.lem.2}
	Given $0< \varsigma\leq 1$ and $f\in \sm{m}$, consider
        $f\circ a$, then there exists a sequence of symbols
        $\{\tilde{f}_{j}\}_{j \in \N}$ with $\tilde f_j \in
        \sm{m-\varsigma j}$ $\forall j$ and $\tilde{f}_0 = f$ such
        that for any $n \in \N$
	\begin{equation}
	\Op(f\circ a) = \sum_{0\leq j< \tN} \tilde{f}_j(A) + R^\prime_\tN\,,
	\end{equation}
	where $R^\prime_\tN$ is a smoothing operator of order
        {$\varsigma \tN- m$}. Moreover,
        \begin{itemize}
        \item[(1)] 
        $\displaystyle{\textrm{supp\,}(\tilde{f}_j) \subseteq
          \textrm{supp\,}(f)}$ $\forall j$

          \item[(2)] \label{controllo} $\exists \alpha,\beta$ and $C_\tN$, independent of $f$ s.t.
          $$\| R'_\tN\|_{s, s - (m-\varsigma \tN)}\leq
          C_\tN\left\| f\right\|^{(m)}_{\alpha,\beta}\,.
          $$
        \end{itemize}
\end{lemma}
In what follows we will denote $\tilde f\sim\sum_{j\geq 0} \tilde
f_j$. 

\begin{proof}
	The proof is obtained arguing inductively. In particular, fix
        $\tN \in \N$: we prove that for any $\bar{n} \leq \tN$ there
        exist $\tilde{f}_0, \dots, \tilde{f}_{\bar n}$, with $\tilde{f}_j \in \sm{m-\varsigma j}$ $\forall j,$ such that 
	\begin{equation}\label{fj}
{\Op(f \circ a)} - \sum_{0\leq j < \bar{n}} \tilde{f}_j(A) = \Op(\tau^{(\bar{n})}{\circ a}) + R^\prime_{\bar{n}}\,,
	\end{equation}
	where $\tau^{(\bar{n})} \in \sm{m-\varsigma (\bar{n}+1)}$ and $R^\prime_{\bar{n}}$ is a smoothing operator of order {$\varsigma \tN-m$}.
	If $\bar{n}=0$, Lemma \ref{cha.lem.1} implies that \eqref{fj} is satisfied with $\tilde{f}_{0} = f$. Indeed, one has
	$$
	{\Op(f \circ a)} - \tilde{f}_0(A) = {\Op(f \circ a)} - f(A) = - \sum_{1\leq j < \tN} \Op(f_j {\circ a}) - R_\tN\,,
	$$
	which implies \eqref{fj} with $\tau^{(0)} := - \sum_{1\leq j < \tN} f_j \in \sm{m-\varsigma}$ and $R^\prime_0 := R_\tN$.\\
	Suppose now that \eqref{fj} is satisfied for some $\bar{n}\geq 0$; then one chooses $\displaystyle{\tilde{f}^{(\bar{n}+1)} = - \tau^{(\bar{n})}}$ and, again by Lemma \ref{cha.lem.1}, obtains
	\begin{align*}
{\Op(f \circ a)} - \sum_{0\leq j < \bar{n} + 1} \tilde{f}_j(A) &= {\Op(f \circ a)} - \sum_{0\leq j < \bar{n}} \tilde{f}_j(A) + \tau^{(j)}(A)\\
	& = \Op(\tau^{(\bar{n})}{\circ a}) + R^\prime_{\bar{n}} - \sum_{ 0 \leq j \leq \tN} \Op(\tau^{(\bar{n})}_{j}{\circ a}) + R_{\bar{n}}\,,
	\end{align*}
	with $R_{\bar{n}}$ a smoothing operator of order {$\varsigma  \tN-m$}. Thus \eqref{fj} is satisfied at the step $\bar{n}+1$, with $R^\prime_{\bar{n}+1} = R^\prime_{\bar{n}} - R_{\bar{n}}$ and
	$\tau^{(\bar{n}+1)} = - \sum_{ 1 \leq j \leq \tN} \tau^{(\bar{n})}_{j}$.
\end{proof}

By the above Lemma $\psi_{\ta}$ is a quasimode for $\Op (f\circ a)$,
so one immediately gets the following
Lemma

\begin{lemma}
  \label{spettri}
Fix $m \in \R$ and $0< \varsigma \leq 1$. Suppose $F = \Op(f)$ is a
self-adjoint operator whose symbol $f\in \Snoi^{m}$ is a function
of the actions only. Then there exists a sequence $\tilde f_{j} \in
\sm{m-\varsigma j}$ of functions with
$\tilde f_0 = {f_c}$ such that, for any $\ta \in \left(\Z^2 + \kappa\right) \cap
\cC$ {sufficiently large}, there exists an eigenvalue $ \lambda_\ta$ of $F$ fulfilling
\begin{equation}\label{7.6}
\lambda_{\ta} \sim \tilde f(\ta) =\sum_{j\geq 0}\tilde f_{j} (\ta)\,.
\end{equation}
\end{lemma}
\proof {  Let $\Psi\in\sm 0$ be a cutoff function equal
  to 1 on $\cC$, with support contained in $\cV$ which is homogeneous
  of degree zero. {By Lemmas \ref{cha.lem.1} \ref{lem.a.dependence},  and \ref{cha.lem.2}, for $\ta\in\tC$ sufficiently large one has}
  \begin{equation}\label{opf.sp}
  \begin{aligned}
    \Op (f) \psi_{\ta}
    &=    \Op (f)(1-\chi(\|A\|)) \Psi(A)\psi_{\ta}
\\    &= \Op\left({f\sharp
 (1-\check \chi(\|a\|))\sharp     \check \Psi(a)}\right)\psi_{\ta}+R_N\psi_{\ta}
    \\
    &= \Op\left({ f_c\sharp
      (1-\check \chi(\|a\|))\sharp     \check\Psi\circ a}\right)\psi_{\ta}+R_N\psi_{\ta}
    \\
&=\Op(f_c {(a)})
(
      (1-\chi(\|A\|)))\Psi(A)\psi_{\ta}+R_N\psi_{\ta}
\\ &=    \Op(f_{c} {(a)}) \psi_{\ta}+R_N\psi_{\ta} = {\widetilde{f_c}}(\ta)\psi_{\ta}+R_N\psi_{\ta}\ ,
  \end{aligned}
  \end{equation}
 with $R_N$ a regularizing operator which changes from line to line.
 From this equation, by a quasimode argument the thesis immediately
 follows (in the statement, we just omitted the index $c$ from
 {$\tilde f$}). \qed

We come to the proof of Theorem \ref{prop.victor.hugo}:
\begin{proof}[Proof of Theorem \ref{prop.victor.hugo}] We apply
  again a quasimode argument to $H_\tN$. {We start with observing that, up to a pseudo-differential operator of order $m-\tN\rho$, $H_\tN$ has symbol given by 
	$$
	f + z^{(res)}_\tN\,, \quad f := {h}_0({a}) + \sum_{j=1}^{\tN} {z}_j({a})\,.
	$$ With the notations of Lemma \ref{spettri},} we aim at proving that, for
  $\ta\in\Omega$ sufficiently large,
	\begin{equation}\label{hN}
	\left\| H_\tN \psi_{{\ta}} - \left(\tilde{h}_0({\ta}) + \sum_{j=1}^{\tN} \tilde{z}_j({\ta})\right) \psi_{{\ta}} \right\| \lesssim_{\tN} \|\ta\|^{m-\tN\rho}\,,
	\end{equation}
so that there exists an eigenvalue $\lambda_{\ta}$ of $H_\tN$ fulfilling
	$$
	\left| \lambda_{\ta} - \tilde{h}_0(\lambda_{\ta}) - \sum_{j=1}^{\tN} \tilde{z}_j(\lambda_{\ta})\right| \lesssim_{\tN} \|\ta\|^{m-\tN\rho}\,.
	$$
	{First we focus on the normal form term $z^{(res)}_\tN$.} Define 
	\begin{equation}\label{psik}
	\eta_k(a) := \chi \left(\frac{\omega(a) \cdot k}{2 \|k\| \|a\|^\delta}\right) \left(1-\chi\left(\frac{\|k\|}{2 \|a\|^{\ep}}\right)\right)\,;
	\end{equation}
	then, arguing as in Lemmas \ref{stima g chi f psi} and \ref{stima tilde} we have that $\eta_k(a) \in \Snoi^{0},$ with seminorms that are uniformly bounded in $k$, and
	\begin{equation}\label{d.1}
	\begin{gathered}
	z^{(res)}_\tN = \sum_{k \neq 0} \hat{z}^{(res)}_{\tN, k} =  \sum_{k \neq 0} \hat{z}^{(res)}_{\tN, k} \eta_k = \sum_{k \neq 0} \left(\hat{z}^{(res)}_{\tN, k}\sharp \eta_k + r^{(2)}_{\tN, k}\right) \,,\\
	r^{(2)}_{\tN, k} \in \Snoi^{-\bar{n}} \quad \forall \bar{n} \in \N\,,
	\end{gathered}
	\end{equation}
	where the second equality in \eqref{d.1} is due to the fact that
	$\eta=1$ on the support of $\hat z^{(res)}_{\tN,k}$, and the {third} equality is due to Lemma \ref{lemma composizione 2}
	and to the fact that also the derivatives of $\eta$ vanish on the support
	of $\hat z^{(res)}_{\tN,k}$.
	
	It follows that 
	\begin{align*}
	\Op(z^{(res)}_\tN)=\sum_{k}\left(\Op(\hat
	z_k)\Op(\eta_k)+R^{(2)}_{\tN,k}\right)
	\\
	=\sum_{k} \left(\Op(\hat z_k)\left(\tilde\eta_k(A)+R^{(3)}_{\tN,k}\right)+ R^{(2)}_{\tN,k}\right)\,,
	\end{align*}
	with $\tilde\eta_k$ a function with the same support of $\eta_k$ and
	with uniformly bounded norms and $R^{(2)}_{\tN,k}$, $R^{(3)}_{\tN,k}$ smoothing
	operator with norms which are respectively summable and uniformly
	bounded in $k$, due to estimates \eqref{stime2} and \eqref{controllo} respectively.
	{By the very definition of the set $\Omega$, one has
	$\eta(A) \psi_\ta = 0 $
	for any $\ta \in \Omega$, thus
	\begin{equation}\label{reg.z}
	z^{(res)}_\tN\psi_{\ta} = R^\prime_\tN \psi_{\ta}
	\end{equation}
	for some smoothing operator $R^\prime_\tN$. Then combining \eqref{reg.z} and equation \eqref{opf.sp}, with $f = h_0(a) + \sum_{j=1}^{\tN} z_j(a),$ we obtain that there exists $R_\tN \in \Snoi^{m-\tN\rho}$ such that
	$$
	H_\tN \psi_{\ta} =  \left(\tilde{h}_0({\ta}) + \sum_{j=1}^{\tN} \tilde{z}_j({\ta})\right) \psi_{{\ta}} + R_{\tN} \psi_{\ta}\,,
	$$
	which implies \eqref{hN}.
	}
\end{proof}

\section{Cardinality estimates}\label{misura} 

\subsection{Nondegenerate homogeneous frequency
  maps}
In this section we prove that the non resonant set $\Omega$ {defined
in \eqref{def.om}} is of density one in $\tC$, {thus concluding the proof of Theorem \ref{maint}.} More precisely, we
prove that the complementary of $\Omega$ has density zero.

The strategy consists in reducing the estimate of the cardinality of
sets to measure estimates. Then we have to estimate the measure of
resonant sets. To this end we exploit the homogeneity of the
nonresonance condition with respect to $a$ in order to reduce the estimate of
their measure to a
measure estimate on the intersection of $\cC$ with the unit
sphere. Finally, the estimate on the unit sphere is done exploiting
the tools developed in the context of degenerate KAM theory, in
particular by R\"ussmann.\\
First of all, define
$$
\Gamma := \Z^2 + \kappa\,.
$$
We start by defining the resonant sets and the ``cutoff sets''.
\begin{align}
  \label{sigma}
\Sigma_k(\gamma):=\left\{ a\in\cC\ :\ \left|\omega(a)\cdot
k\right|\leq \gamma \|k\|\|a\|^{\delta}\right\}
\\
\label{tks}
\cT_k(\gamma):=\left\{a\in\cC \ :\ \|k\|< {\gamma}\|a\|^{\epsilon}\right\}
\ ,
\\
\label{sigma2}
\Sigma(\gamma):=\bigcup_{k\in\Z^2\setminus\left\{0\right\}}
\left(\Sigma_k(2)\cap\cT_k(2)\right) 
\ .
\end{align}
The main result of this section is the following Theorem. 
\begin{theorem}
  \label{main.mea}
  Assume $\delta>M-1$, then $\exists C>0$, $\mu_0\in\N$ s.t. for $R$
  large enough, one has 
  $$
\#\left(\Gamma\cap\Sigma(\gamma)\cap B_R\right)\leq
C\gamma^{1/\mu_0}\frac{R^2}{R^{\frac{M-\delta}{\mu_0}-2\epsilon}} \ .
  $$
\end{theorem}
The rest of the section is devoted to the proof of this theorem.

First, we fix a large $R_0$ and we will work in the ball in the action
space {centered at the origin and} with radius $R$ larger than $R_0$: $R\geq R_0$. In the following
we will assume that $R_0$ is as large as needed. Following \cite{BLMnr}, given $r \in\R^+$ and a set $\cA$, we define
$$
\cA^{(r)}:=\bigcup_{{\ta}\in\cA}B_r({\ta})\ ,
$$
so that we have the following remark
\begin{remark}
  \label{palleenumeri}[Remark 5.12 of \cite{BLMnr}] Let $\cA$ be a
  set and let $r<1/2$, then
$$
\#(\cA\cap\Gamma) \leq \frac{\left|\cA^{(r)}\right|}{\left|{B_r}\right|}\ .
$$
\end{remark}

\begin{lemma}
  \label{tkss}
  $$
\cT^{(r)}_k(2)\subset \cT_k\left(1\right)\ .
  $$
\end{lemma}
\proof By definition
$$
\cT^{(r)}_k(2)=\left\{a\ : \exists \tilde a\ :\ \|a-\tilde a\|\leq
r\ ;\ \|\tilde a\|^\epsilon\geq 2 \|k\| \right\}\ .
$$
We study $\|a\|$. One has
$$
\left\|a\right\|\geq \left\|\tilde a\right\|-\|a-\tilde a\|\geq
2 \|k\|^{1/\epsilon}-r=\|k\|^{1/\epsilon}\left(2-\frac{r}{\|k\|^{1/\epsilon}}\right)
\ ,
$$
but the parenthesis is larger than $1$, as it is easy
to verify using $r< 1/2$.\qed

\begin{lemma}
  \label{sigmar}
  Define
  $$
C:=\sup_{\|a\|=1,a\in\cC}\left\|{d\omega(a)}\right\|\ ,
  $$
assume 
\begin{equation}
  \label{sua}
\|a\|\geq
 \frac{1}{2}\left(\frac{2rC}{\gamma}\right)^{\frac{1}{\delta+1-M}}
\end{equation}
then
$$
\Sigma_k^{(r)}(\gamma)\subset
\Sigma_k\left(\tilde{\gamma}\right)\,, \quad \tilde \gamma := \frac{\gamma}{2^{\delta+1}}\ .
$$
\end{lemma}
\proof We denote $a=\lambda u$, with $\lambda=\|a\|$ and $u\in\cC$, and
similarly $\tilde a=\tilde \lambda \tilde u$ and so on. Let $a\in
\Sigma_k^{(r)}(\gamma)$, then there exists $\bar a\in
\Sigma_k(\gamma)$ s.t. $\left\|a-\bar a\right\|\leq r$, thus one has
$$
\left|\frac{\omega(a)\cdot k}{\|k\|}\right|\geq
\left|\frac{\omega(\bar a)\cdot k}{\|k\|}\right| -
\left|\frac{d(\omega(\tilde a)\cdot k)(a-\bar a)}{\|k\|}\right| \ , 
$$
with some $\tilde a$, fulfilling $\frac{\|a\|}{2}\leq \|\tilde a \|\leq
2\|a\|$. Of course the same inequality is true if one replaces $\bar a$
to $\tilde a$. So we have  
\begin{equation}
  \label{sua.2}
\left|\frac{d(\omega(\tilde a)\cdot k)(a-\bar a)}{\|k\|}\right|\leq
C\tilde\lambda^{M-1} r\leq C 2^{M-1}\lambda^{M-1} r \ ,
\end{equation}
and also 
\begin{equation}
  \label{sua.1}
\left|\frac{\omega(\tilde a)\cdot k}{\|k\|}\right| \geq
\tilde \lambda^{\delta}\gamma\geq \frac{\lambda^{\delta}\gamma
}{2^\delta}\ .
\end{equation}
If \eqref{sua} is satisfied then \eqref{sua.2} is smaller than a half
of \eqref{sua.1} and implies
$$
\left|\frac{\omega(a)\cdot k}{\|k\|}\right|\geq \frac{\lambda^{\delta}\gamma
}{2^{\delta+1}}\ ,
$$
which is the thesis. \qed

We are now going to estimate the measure of
$\Sigma_k(\gamma/2^{\delta+1})\cap\cT_k(1)\cap B_R$. To this end we
exploit the homogeneity of the frequencies. We denote
\begin{equation}
  \label{slambda}
S_\lambda:=\left\{a\in\R^2\ : \ \|a\|=\lambda\right\}\ ,
\end{equation}
and we will exploit the following  
\begin{remark}
  \label{integro}
  $$
\left|\Sigma_k(\tilde \gamma)\cap\cT_k\left(1\right)\cap
B_R \cap B_{R_0}^{c}\right| =\int_{R_0}^R \left|\Sigma_k(\tilde \gamma)\cap\cT_k\left(1\right)\cap
S_\lambda\right|d\lambda\ .
$$
\end{remark}
In order to estimate the above quantity we establish the following
Lemma
\begin{lemma}
  \label{il nostro}
  $$
\Sigma_k(\tilde \gamma)\cap\cT_k\left(1\right)\cap
S_\lambda=\left\{\begin{matrix}
\emptyset & {\rm if}\ \lambda^\epsilon\leq\|k\|
\\
\lambda\left(\Sigma_k\left(\frac{\tilde
  \gamma}{\lambda^M-\delta}\right)\cap S_1\right)& {\rm
  if}\ \lambda^\epsilon> \|k\|
\end{matrix}\right.\,,
$$
where the multiplication of a set by $\lambda$ means multiplication of
each one of its elements. 
\end{lemma}
\proof Just remark that (with $u=a/\lambda$), we have
$$
\Sigma_k(\tilde \gamma)\cap S_\lambda=\left\{a=\lambda u \ : u\in S_1
\ and\ \left|\frac{\omega(\lambda u)\cdot k}{\|k\|}\right|\geq
                  {\lambda^{\delta}\tilde \gamma
}   \right\}\ ,
$$
but the nonresonance condition can be rewritten using the homogeneity
of $\omega$ as 
$$
\left|\frac{\omega(u)\cdot k}{\|k\|}\right|\geq
           {\lambda^{\delta-M}\tilde \gamma
} \ .
$$
In order to conclude the proof just remark that the intersection with
$\cT_k\left(1\right)$ is empty or full according to the
conditions in the Lemma. \qed

We are now in the position of using degenerate KAM theory in order to
estimate $\Sigma_k\left(\frac{\tilde
  \gamma}{\lambda^M-\delta}\right)\cap S_1$. First we recall a couple
of lemmas and definitions from \cite{Rus01} (see also \cite{BBM11}).

First we adapt the notation. Thus, consider the functions
$(\omega_1(a),\omega_2(a))$, and restrict them to the intersection of
{the set $\Pi$} with the unit sphere. Precisely, 
consider
\begin{align}
  \label{omphi}
\omega(\phi)\equiv(\omega_1(\phi),\omega_2(\phi)) 
\\
\label{onphi.2}
\omega_j(\phi):=\omega_j(\cos\phi,\sin\phi)\ ,\quad
{\phi\in\left[0,\frac{3}{4}\pi\right]}\ .
\end{align}
\begin{definition}
  \label{nondeg}
The function $(\omega_1(\phi),\omega_2(\phi))$ is said to be
{weakly} nondegenerate if $\forall (c_1,c_2) \not=(0,0)$ the function
$$
c_1\omega_1(\phi)+c_2\omega_2(\phi)
$$
is not identically zero. 
\end{definition}

We will prove below (see Lemma \ref{ferqnondeg}) that $\omega(\phi)$
is analytic on a complex neighborhood of the interval ${[0,3/4\pi]}$
and that it is weakly nondegenerate.

\begin{lemma}
  \label{rus2}
Assume that $\omega$ is weakly nondegenerate, then there exist
$\beta>0$ and $1\leq \mu_0\in\N$ s.t.
\begin{equation}
  \label{basso}
\max_{0\leq\mu\leq\mu_0}\left|\frac{d^\mu}{d\phi^\mu}\frac{k\cdot
  \omega(\phi)}{\|k\|}\right|\geq\beta\ ,\quad \forall
\phi\in{[0,3/4\pi]}\ ,\quad \forall k\in\Z^2\setminus\left\{0\right\}\ .
\end{equation}
\end{lemma}
\proof By contradiction: assume that $\forall \mu_0$ and 
{$\forall \beta >0$} $\exists \phi_{\mu_0,\beta}$,
$k_{\mu_0,\beta}$ s.t.
$$
\max_{0\leq\mu\leq\mu_0}\left|\frac{d^\mu}{d\phi^\mu}\frac{k_{\mu_0,\beta}\cdot
  \omega(\phi_{\mu_0,\beta})}{\|k_{\mu_0,\beta}\|}\right|<\beta\ .
$$
Take $\sigma:=\mu_0$, $\beta:=(\sigma+1)^{-1}$, then $\exists
\phi_\sigma,k_\sigma$ s.t.
$$
\max_{0\leq\mu\leq\mu_0}\left|\frac{d^\mu}{d\phi^\mu}\frac{k_\sigma\cdot
  \omega(\phi_{\sigma})}{\|k_\sigma\|}\right|<\frac{1}{\sigma+1}\ .
$$
But this means that $\forall \mu$, $\exists \sigma\geq\mu$ s.t.
\begin{equation}
  \label{lasigma}
\left|\frac{d^\mu}{d\phi^\mu}\frac{k_\sigma\cdot
  \omega(\phi_{\sigma})}{\|k_\sigma\|}\right|<\frac{1}{\sigma+1}\ .
\end{equation}
Take the limit $\sigma\to\infty$. By compactness $\phi_\sigma\to\bar
\phi$ and $\frac{k_\sigma}{\|k_\sigma\|}\to \bar c= (\bar c_1,\bar c_2)$. Thus taking the
limit of \eqref{lasigma}, one gets
$$
\left|\frac{d^\mu}{d\phi^\mu}{\bar c\cdot
  \omega(\bar \phi)}\right|=0\ .
$$
But, by analyticity, this means $\bar c\cdot
  \omega(\bar \phi)\equiv 0$, against the assumption of weakly
  nondegeneracy. \qed

  We now recall the following theorem which is a simplification of a
  theorem by R\"ussmann. For the very
  technical proof we make reference to the original paper
  \begin{theorem}
    \label{rus}[Theorem 17.1 of \cite{Rus01}]
    Let $\cI\subset\R$ be compact and denote by $|\cI|$ its
    length. Denote (as above)
    $$
\cI^{(r)}:=\bigcup_{\phi\in\cI}B_r(\phi)\ .
$$
Let $g\in C^{\mu_0+1}(\cI^{(r)})$ be s.t.
$$
\min_{\phi\in\cI}\max_{0\leq\mu\leq\mu_0}\left|\frac{d^\mu
  g}{d\phi^\mu}(\phi)\right|\geq \beta\ .
$$
Then {$\forall \ep>0$}
\begin{equation}
  \label{stimie}
\left|\left\{\phi\in\cI\ :\ |g(\phi)|\leq \epsilon\right\}\right| \leq
C
|\cI|\left(\frac{\epsilon}{\beta}\right)^{\frac{1}{\mu_0}}\frac{1}{\beta}\left|
g\right|_{C^{\mu_0+1}(\cI^{(r)})}\ .
\end{equation}
  \end{theorem}
In the original version the constant $C$ is explicitly computed, but
here we do not need its value. 

\begin{corollary}
  \label{suSfere}
  $\exists \mu_0$, $C>0$ s.t.
  \begin{equation}
    \label{suSfere.1}
\left|\Sigma_k(\gamma)\cap S_1\right|\leq C\gamma^{1/\mu_0}\ , \forall
k\in\Z^2\setminus\left\{0\right\}\ . 
  \end{equation}
\end{corollary}

Using this corollary we prove now the following Lemma.
  
\begin{lemma}
  \label{misura.1}
  There exist  $C>0$ such that, if $R>R_0$,
  \begin{align}
    \label{mes}
\left|\bigcup_{k\not=0}\left(\Sigma_k(\tilde \gamma)\cap\cT_k\left(1\right)\right)\cap
B_R\right| \leq C\frac{R^2}{R^{\frac{M-\delta}{\mu_0}-2\epsilon}}
  \end{align}
\end{lemma}
\proof Define $B_{R, R_0} := B_R \cap B_{R_0}^c$ and just observe that
$$
\begin{aligned}
\left|\bigcup_{k\not=0}\left(\Sigma_k(\tilde \gamma)\cap\cT_k\left(1\right)\right)\cap
B_R\right|
&=\left|\bigcup_{k\not=0}\left(\Sigma_k(\tilde \gamma)\cap\cT_k\left(1\right)\right)\cap
B_{R_0}\right| \\
&+ \left|\bigcup_{k\not=0}\left(\Sigma_k(\tilde \gamma)\cap\cT_k\left(1\right)\right)\cap
B_{R, R_0}\right|\\
& \leq 2  \left|\bigcup_{k\not=0}\left(\Sigma_k(\tilde \gamma)\cap\cT_k\left(1\right)\right)\cap
B_{R, R_0}\right|
\end{aligned}
$$
if $R$ is big enough, thus
\begin{align}
  \label{mea.1}
\left|\bigcup_{k\not=0}\left(\Sigma_k(\tilde \gamma)\cap\cT_k\left(1\right)\right)\cap
B_{R}\right|
&\leq 2\int_{R_0}^R
\left|\bigcup_k\left(\Sigma_k(\tilde\gamma)\cap\cT_k\left(1\right)\cap
S_\lambda\right)\right|d\lambda\\
&= 2 \int_{R_0}^R
\left|\bigcup_{|k|\leq 2\lambda^{\epsilon}}\left(\Sigma_k(\tilde\gamma)\cap
S_\lambda\right)\right|d\lambda
\\
&= 2
\int_{R_0}^R\lambda
\left|\bigcup_{|k|\leq \lambda^{\epsilon}}
\left(\Sigma_k(\frac{\tilde\gamma}{\lambda^{M-\delta}})\cap 
S_1\right)\right|d\lambda
\\
&\leq 8 \int_{R_0}^R \lambda^{1+2\epsilon}
\sup_{k}\left|
\left(\Sigma_k(\frac{\tilde\gamma}{\lambda^{M-\delta}})\cap 
S_1\right)\right|d\lambda
\\
&\leq C\int_{R_0}^R
\lambda^{1+2\epsilon}\left(\frac{\tilde
  \gamma}{\lambda^{M-\delta}}\right)^{1/\mu_0} d\lambda
\\
&\leq C\frac{R^2}{R^{\frac{M-\delta}{\mu_0}-2\epsilon}}\,.
\end{align}
\qed
{\begin{remark}\label{mu.legame}
By Lemma \ref{misura.1} one immediately deduces that the set $\tilde \Omega$ defined in \eqref{def.om.2} has density one at infinity, provided $\ep$ and $\delta$ are such that $2\ep \mu_0< {M-\delta}$. Thus the same holds true also for the set $\Omega = \tilde{\Omega} \cap B_{\tt R}^c$ defined in \eqref{def.om}.
\end{remark}
}
\subsection{Nondegeneracy of the frequency map of the anharmonic
  oscillator} \label{nondeg}

In this section we prove that the restriction of the frequency map to
the unit sphere is
weakly nondegenerate.

Actually the proof is a simple consequence of
the techniques developed in \cite{Fej04,BF18}), here we give it for
the sake of completeness.

\begin{lemma}
  \label{ferqnondeg}
The function $(\omega_1,\omega_2)$ is analytic in a neighborhood
  of $[0,3/4\pi]$ and is nondegenerate.
\end{lemma}
\proof Analyticity in the open set {$(0,3/4\pi)$} follows from the fact
that the Hamiltonian is an analytic function of the actions on the
interior of $\Pi$. Then remark that the points {$0,3/4\pi$ correspond to
$a_r=0$,} which is the circular orbit. {We consider just the case
  of $a_2>0$, the other one follows easily.}
$\omega(0)$ {is the
limit as $\phi\to 0$} of $\omega(\phi)$. Now, $\omega_2$ just
converges to the frequency of the circular orbit,
while $\omega_1$ converges to the frequency of small oscillation of
the effective system with Hamiltonian
\begin{equation}
  \label{effham}
h_0(r,p_r,a_2):=\frac{{p_r^2}}{2
}+\frac{a_2^2}{2r^2}+\frac{r^{2\ell}}{2\ell} \ .
\end{equation}
Now it is well known that the Birkhoff normal form at the minimum of
the effective potential (namely the circular orbit) is convergent, and
this allows to compute the development of the Hamiltonian, as a
function of the actions, at the circular orbit. In particular it turns
out that also the frequencies can be extended to complex analytic
functions close to the circular orbit. 

It was proved in \cite{Fej04} that close to the
circular orbit the Hamiltonian admits the expansion (for $a_1\ll |a_2|$) 
\begin{equation}\label{HamExp}
h_0(a_1,a_2)=V^*(a_2)+\sqrt{A(a_2)}a_1+\frac{-5B(a_2)^2+3C(a_2)A(a_2)}{48A(a_2)^2}a_1^2+o(a_1^2)\; ,
\end{equation}
where
$$\begin{aligned}
&V^*(a_2)=\frac{a_2^2}{2r_0^2}+V(r_0) 	\; , &A(a_2)=\frac{3 a_2^2}{r_0^4}+V''(a_0) \; ,\\
&B(a_2)=-\frac{12 a_2^2}{r_0^5}+V'''(r_0) \; , & C({a_2})=\frac{60
    a_2^2}{r_0^6}+V^{(4)}(r_0)\; ,\\
  &\quad r_0:= a_2^{\frac{1}{\ell+1}}\ .
\end{aligned}$$
We recall that, to obtain this formula, one has to think of $a_1$ as $\frac{\tilde p_r^2+\tilde
  r^2}{2}$ where $\tilde p_r$ and $\tilde r$ are rescaled, translated
variables. An explicit computation gives
\begin{equation}
  \label{hinazang}
h_0=\frac{\ell+1}{2\ell}a_2^{\frac{2\ell}{\ell+1}}+\sqrt{2(\ell+1)}a_2^{M}a_1+\frac{1}{2}\frac{3ac-5b^2}{24a^2}a_2^{-\frac{2}{\ell+1}}a_1^2+h.o.t. 
\end{equation}
where
$$
a=2(\ell+1)\ ,\quad b:=2(2\ell^2-2\ell-5)\ ,\quad
c:=2(4\ell^3-12\ell^2+11\ell+27)\ .
$$
(We emphasize that these letters will be used to denote such quantities
\emph{only in this proof}).
From \eqref{hinazang}, one can compute
$\omega_j=\frac{\partial h}{\partial a_j}$ close to any point of the
form $(a_1,a_2)=(0,\bar a_2)$, from which one sees that it extends to
a complex analytic function in a neighborhood of such point. One can
also compute
$\omega_2/\omega_1$, getting (with some work)
\begin{align}
\label{fra0}
\omega_1&\simeq \sqrt{2(\ell+1)}a_2^M\left(1+\frac{3ac-5b^2}{24a^2\sqrt{2(\ell+1)}  }\frac{a_1}{a_2}\right)
\ ,
\\
\label{fra00}
\omega_2&\simeq a_2^M\left(1+M\sqrt{2(\ell+1)}\frac{a_1}{a_2}\right)\ ,
\\
  \label{fra1}
\frac{\omega_2}{\omega_1}&\simeq
\frac{1}{\sqrt{2(\ell+1)}}\left(1+\left(M\sqrt{2(\ell+1)}-\frac{3ac-5b^2}{24a^2\sqrt{2(\ell+1)}}
\right)\frac{a_1}{a_2}\right) \ .
\end{align}
In particular, from \eqref{fra0} and \eqref{fra00}, we notice that
both $\omega_1$ and $\omega_2$ do not vanish
for $a_2\not=0$ and $0<a_1\ll|a_2|$.

A further computation gives the value of the coefficient of $a_1/a_2$
in \eqref{fra1},
which turns our to be
$$
d:=\frac{16(\ell+1)^2(\ell-1)(2\ell+1)}{24a^2\sqrt{2(\ell+1)}}\ .
$$
So, such a coefficient vanishes only if $$
\ell=-1/2,1,-1\ .
$$
By the way, we remark that they correspond to the Kepler and the
Harmonic cases, and also to the degenerate case of potential
proportional to $r^{-2}$.

So, if $c_2\not=0$, $c_1\omega_1+c_2\omega_2$ vanishes only if
$$
c_1+c_2\left(\frac{1}{\sqrt{2(\ell+1)}  }\left(
1+d\frac{a_1}{a_2}\right)  \right)  =0\ ,
$$
but this is a nontrivial analytic function of $a_1/a_2$, therefore it
is also a nontrivial function of $\phi$. If $c_2=0$ then $c_1\not=0$,
and therefore the same is true, since $\omega_1$ is different from
zero.
\qed

\appendix

\section{Proof of Lemma \ref{azioni}}\label{sec.residui}

We start by proving the following easy lemma

\begin{lemma}
  \label{Lnonzero}
  The function $a_r$ defined by \eqref{ar} is analytic on the domain
    \begin{equation}
      \label{defia2}
0<\left|L\right|<\left(\frac{2\ell}{\ell+1}E\right)^{\frac{\ell+1}{2\ell}}\ ,\quad
E>0\ ;
    \end{equation}
\end{lemma}
\proof The effective Hamiltonian $h_0^*(r,p_r,L)$ defined by
\eqref{efficace} is an analytic function except at $r=0$. Furthermore,
for $L\not=0$, $h^*$, as a function of $r,p_r$ is a submersion, except
on the level surface of level $E=|L|^{\frac{\ell+1}{2\ell}}$. So,
outside this domain the level surface depend analytically on both $E$
and $L$. Since $a_r$ is just the normalized area contained in the
level surface, it also depends analytically on $E,L$ in the considered
domain. \qed

We now study the behavior of $a_r$ as $L\to 0$. To this end we apply the
method of the residues to the integral defining it. Precisely, we
prove the following Lemma


\begin{lemma}\label{lemma.residui}
Define 
\begin{equation}
  \label{rho}
\rho=: \frac{\azdue^2}{E^{\frac{(\ell+1)}{\ell}}}\ ,
\end{equation}
then there exists $\rho_*$ and a function $f(E,L)$, analytic in the
domain
\begin{equation}
  \label{doma}
E>0, \quad 0\leq \rho<\rho_*\ ,
\end{equation}
s.t. 
  \begin{equation}\label{holomor}
		a_r = -\frac{1}{2} |\pth| + f(E, \pth)\,.
	\end{equation}
\end{lemma}
\begin{proof}
Performing the change of variables $r^2 =s,$ in the integral defining
$a_r$ one has
\begin{align} 
\azuno = \frac{\sqrt{2}}{2 \pi} \int_{s_m}^{s_M} \frac{1}{s}\sqrt{-p(s, E, \azdue)}\ d s\,, \label{cdv}
\end{align}
with
\begin{equation}\label{def.p}
{p(s, E, \azdue) = \frac{s^{\ell+1}}{2\ell} - Es + \frac{\azdue^2}{2}}
\end{equation}
and $s_m = s_m(E, \azdue),$ $s_M=s_M(E, \azdue)$ the two positive solutions of the equation ${p(s, E, \azdue) = 0\,.}$\\

We now study the zeroes of $p(s, E, \azdue)$. To this aim we observe
that, since $p$ is a polynomial in $s$ with coefficients depending on
$E$ and $\azdue$, it has $\ell+1$ complex roots which of course depend
continuously on $E$ and $L$. Actually there is more structure: indeed
$s_*(E,L)$ is a root of $p$ if and only if $t_*=s_*E^{-1/\ell}$ is a
root of
\begin{equation}\label{def.rho}
\tilde p(t) = \frac{t^{\ell+1}}{2\ell} - t + \frac{1}{2}\rho\,, \quad \,,
\end{equation}
which is a function of $\rho$ only. 

Since we are interested in a neighborhood of $\rho=0$, we start by
remarking that, for $\rho=0$ the roots of $\tilde p$ are 
\begin{equation}\label{ell.zero}
t_0(0) = 0\,, \quad \left(t_j(0)\right)^{\ell} = 2 \ell \quad \forall j = 1, \dots, \ell\,.
\end{equation}
To be determined we take $t_1(0):=(2\ell)^{1/\ell}$ to be the
positive real root of $2\ell$ and the other roots in counterclockwise
order. Correspondingly we will have
$$
E^{1/\ell}t_0(\rho)=s_m(E,L)\ ,\quad E^{1/\ell}t_1(\rho)=s_M(E,L)\ .
$$
Denote
$$
d:=\min_{j\not=l}\left|t_j(0)-t_l(0)\right|\ , 
$$
then there exists $\rho_*$ s.t., for $\rho<\rho_*$ one has
$$
\min_{j\not=l}\left|t_j(\rho)-t_l(\rho)\right|\geq d/2
$$
and correspondingly
\begin{equation}\label{dist.zero.3}
\min_{j\not=l}\left|s_j(E, \azdue) - s_{l}(E, \azdue)\right| \geq
\frac{d}{2} E^{\frac{1}{\ell}}
\end{equation}
The function to be integrated in \eqref{cdv} is 
\begin{equation}\label{non.sono.olomorfa}
F(z) := \frac{\sqrt{2}}{2\pi}\frac{1}{z} \sqrt{- p(z, E, \azdue)}\,,
\end{equation}
In order to make it holomorphic we cut $\C$ along the segments $b_j$
joining $s_{2(j-1)}(E, \azdue)$ with $s_{2 j - 1}(E, \azdue)$,
$j=1,...,\lceil \frac{\ell}{2} \rceil$; if $\ell$ is odd, there is a
last cut $b_{\lceil \frac{\ell}{2} \rceil}$, which is the half-line
parallel to the real axis joining $s_{\ell}(E, \azdue)$ with
$\infty$. Remark that $b_1$ is the interval of integration in which we
are interested. 

We are now ready to choose the curve over we integrate to apply the
method of the residue. To this end we define  
\begin{equation}\label{def.max}
{\tt M} := \max_{j=1,..,\ell, \rho<\rho_*}\textrm{ Re}
\left(t_j(\rho)\right)\ ,
\end{equation}

\begin{figure}
	\centering \includegraphics[scale=0.5]{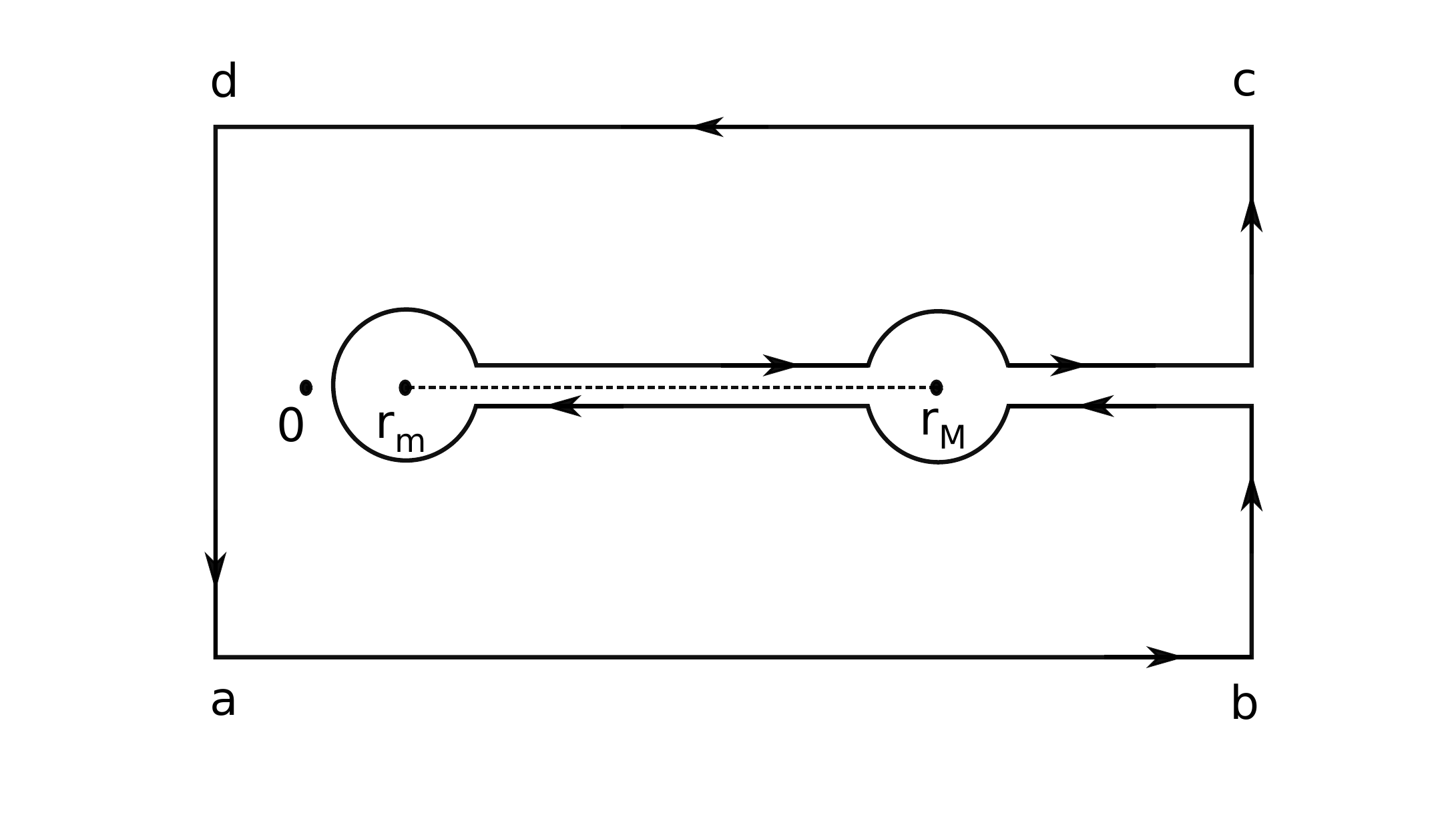}
	\caption{ \small Representation of the path $\Gamma_\varep$. The radius of the circles is $\varep$, and the points $a, b, c, d$ are defined as in \eqref{endpoints}.}
\label{cammino}
\end{figure}

and take the curve $\Gamma_\varepsilon$ described in Figure 1, with
\begin{equation}\label{endpoints}
\begin{gathered}
a = \left(-\frac{d}{4}- \im \frac{d}{4}\right) E^{\frac{1}{\ell}}\,, \quad b= \left(2{\tt M} - \im \frac{d}{4} \right)E^{\frac{1}{\ell}}\,,\\
c =\left(2 {\tt M} + \im \frac{d}{4}\right) E^{\frac{1}{\ell}}\,, \quad d = \left(-\frac{d}{4} + \im \frac{d}{4}\right) E^{\frac{1}{\ell}}\,.
\end{gathered}
\end{equation}
We also denote $\gamma_R$ the boundary of the rectangle $abcd$. With
this notation we have 
	\begin{equation}\label{lim.ep}
\begin{aligned}
2 a_r &= \lim_{\varep \rightarrow 0} \int_{\Gamma_\varep} F(z)\ d z  -  \int_{\gamma_R} F(z)\ d z\\
&= 2 \pi \im \textrm{Res}(F, 0) -  \int_{\gamma_R} F(z)\ d z=-|\azdue|  - \int_{\gamma_R} F(z)\ d z\,.
\end{aligned}
\end{equation}
        To get the thesis just define
        $$
f(E,L):= -  \int_{\gamma_R} F(z)\ d z\ ,
        $$
        and remark that this is analytic in the considered
        domain.
\end{proof}

Joining the results of these two lemmas one gets

\begin{corollary}\label{olomorfa.ovunque}
  There exists function $f=f(E,L)$ analytic in the domain
  \begin{equation}
    \label{domo}
\left\{(E, \pth)\ \left|\ E>0,\ |\pth| < \left(\frac{2\ell}{\ell +
    1} E\right)^{\frac{\ell+1}{2 \ell}}\right.\right \}\ ,
  \end{equation}
  such that
 \begin{equation}\label{ovunque}
  a_r(E, \pth) = -\frac{1}{2}|\pth| + f(E, \pth)\ .
 \end{equation}
\end{corollary}

We finally prove Lemma \ref{azioni}:
\begin{proof}[Proof of Lemma \ref{azioni}]
 Let
 \begin{equation}\label{a1.quella.giusta}
 a_1 := \begin{cases}
 a_r & \textrm{if } \pth \geq 0\\
 a_r - L &\textrm{if } \pth < 0\,.
 \end{cases}
 \end{equation}
then, 
by Corollary \ref{ovunque}, there exists a function $f$ analytic in
the domain \eqref{domo} s.t.
$$
a_1=-\frac{L}{2}+f(E,L)\ ,
$$
so that $a_1$
itself is analytic in such a
domain. This concludes the proof of Items $(1)$ and $(2)$.

We come to the homogeneity properties. Take $L\not=0$, then by
\eqref{ar} and performing the change of variables $r =
|\pth|^{\frac{1}{\ell+1}} x$ in \eqref{ar}, one sees that
 $$
 a_r = \frac{\sqrt{2 E}}{\pi} |\pth|^{\frac{1}{\ell+1}} F(\rho)\,, \quad F(\rho) = \int_{x_m}^{x_M} \sqrt{1-\rho\left( \frac{1}{2x^2} + \frac{x^{2\ell}}{2\ell}\right)}\,d x\,,
 $$
 with $\rho$ as in \eqref{def.rho} and $x_m, x_M$ solving the equation $
 \displaystyle{\rho \frac{x^{2\ell}}{2\ell} + \rho\frac{1}{2x^2} -1 = 0.}$
 As a consequence, for any $\lambda >0$ $a_r$ satisfies
 $$
 a_r(\lambda^{\frac{2\ell}{\ell+1}}E, \lambda \pth) = \lambda^{\frac{\ell}{\ell+1}} \sqrt{2 E} \lambda^{\frac{1}{\ell+1}} |\pth| F(\rho) =\lambda a_r(E, \pth)\,,
 $$
also $a_1$ defined as in \eqref{a1.quella.giusta} satisfies
 \begin{equation}\label{a1.omog}
  a_1(\lambda^{\frac{2\ell}{\ell+1}}E, \lambda \pth) = \lambda a_r(E, \pth)\,.
 \end{equation}
 Since $h_0$ and $a_2$ as in \eqref{L} are quasi-homogeneous functions
 of $x, \xi$, one immediately deduces the quasi-homogeneity of $a_1$.
 The homogeneity of $h_0$ as a function of $a$ also immediately follows
 from \eqref{a1.omog}. We still have to consider the case $L=0$. In
 this case the result immediately follows by continuity, by taking the
 limit of the functions to this set.
 \end{proof}

\section{Proof of Lemma \ref{lem.a.dependence}}\label{tutto.sbagliato}

First we prove the following
\begin{lemma}
  \label{aaa.1}
Let {$\Pi^*:=a^{-1}(\mathop{\Pi}^\circ)$ and suppose} $f\in
C^{\infty}(\R^2)$ is such that $f\circ \phi^{\vf}_a=f$, $\forall \vf$
then there exists $\tilde f\in C^{\infty }(\Piz)$ s.t. $f=\tilde
f\circ a$ { on $\Pi^*$}. 
\end{lemma} 
\proof Introducing action angle coordinates
$(a,\vf)$, which by the standard theory (see e.g. \cite{Dui80}) are
smooth and globally defined on the set
{$\Pi^*$,} the function $f$ is a $C^\infty$
function of $(a,\vf)$ which however does not depend on $\vf$. This
is the wanted function $\tilde f$. \qed

\begin{lemma}
	\label{angolo}
	There exist angle variables $\vphi$ that are quasi-homogeneous of degree $0$ as functions on $\Pi^*$, namely they satisfy
	$$
	\vphi(\lambda x, \lambda^\ell\xi) = \vphi(x, \xi) \quad \forall (x, \xi) \in \Pi^*,\ \forall \lambda \in \R^+\,.
	$$ 
\end{lemma}
\begin{proof}
	Let ${\vphi} = (\vphi_1, \vphi_2): \Pi^* \rightarrow \T^2$ be
        such that $(a, {\vphi})$ are global action angle
        coordinates. For any choice of $\bar a\in \Piz$ with $\|\bar
        a\|=1$,         let $(x_0, \xi_0) \in \Pi^*$  be such that
        $a(x_0,\xi_0)=\bar a$ and $\vphi(x_0, \xi_0) = 0$. Remark that it
        exists because the action $\phi_a^{\vf}$ on the level surfaces
        of $a$ is transitive. Furthermore $(x_0,\xi_0)$ is a function of $a$
        only.  
For $(x, \xi) \in \Pi^*$, define $\lambda \in \R^+$ and $(\widetilde
x,\widetilde \xi)$ by
	\begin{equation}\label{lambda}
	(x, \xi) = (\lambda \tilde{x}, \lambda^\ell \tilde{\xi})\,, \quad \textrm{and}\quad \|a(\widetilde{x}, \widetilde{\xi})\| = 1\,.
	\end{equation}
	Observe that with this definition $\lambda$ is a function of the actions $a$ only. Then this implies that the function $\tilde{\vphi} = (\tilde{\vphi}_1, \tilde{\vphi}_2): \Pi^* \rightarrow \T^2$ given by
	\begin{equation}\label{def.hom}
	\tilde{\vphi}(x, \xi) := \vphi(x, \xi) - \vphi(\lambda
        {x_0}, \lambda^{\ell} {\xi_0})
	\end{equation}
	still defines angle coordinates conjugated to the actions $a$
        on the set $\Pi^*$. This is due to the fact that $\vphi(\lambda
        {x_0},\lambda^\ell\xi_0)$ is a function of the actions only.
        Remark also that
        $$
\tilde \vf(\lambda x_0,\lambda^\ell \xi_0) =0\ ,\quad \forall
\lambda>0\ .
        $$
From now on we use only the angles $\tilde\vf$, so {\bf we omit the
  tildes}. 

        We are now going to prove that these angles are homogeneous
        functions of degree $0$ on $\Pi^*$.

         To this aim, we define for $\mu>0$ and $j = 1,2$
	$$
	\vphi_{\mu, j}(x, \xi):= \vphi_j(\mu x, \mu^{\ell}\xi)
	$$
	First of all, we observe that since $(a, \vphi)$ are
        canonically conjugated variables, one has
	$\displaystyle{\poisson{a_j}{\vphi_i} = \delta_{i, j}}$, so
        one has 
	$$
	\begin{aligned}
	\poisson{a_i}{\vphi_{\mu, j}}(x, \xi) = \mu^{\ell}\partial_{x} a_i (x, \xi) \cdot \partial_{\xi} \vphi_{j} (\mu x, \mu^{\ell}\xi) - \mu\partial_{\xi} a_i (x, \xi) \cdot \partial_{x} \vphi_{j} (\mu x, \mu^{\ell} \xi)\,.
	\end{aligned}
	$$
Since $a_i$ is quasi-homogeneous of degree $\ell+1,$ $\forall (x, \xi) \in \Pi^*$ one has
	$$
	\begin{gathered}
	\mu^{\ell}\partial_{x} a_i (x, \xi)  =  \partial_{x} a_i(\mu x, \mu^{\ell} \xi) \,, \\ \mu \partial_{\xi} a_i (x, \xi)  =  \partial_{\xi}a_i(\mu x, \mu^{\ell} \xi)\,,
	\end{gathered}
	$$
	one also obtains
	\begin{equation}\label{poisson}
	\begin{aligned}
	\poisson{a_i}{\vphi_{\mu, j}}(x, \xi)
	= \poisson{a_i}{\vphi_{ j}}(\mu x, \mu^{\ell}\xi) \equiv
        \delta_{i,j} \ .
	\end{aligned}
	\end{equation}
	Thus, using action angle coordinates to compute Poisson
        Brackets, one has
        $$
\left\{a_j;\vf_{\mu,j}\right\}\equiv\frac{\partial
  \vf_{\mu,j}}{\partial \vf_j} =1\ ,
$$ therefore there exist functions $f_{\mu, j}(a)$, depending on the
actions only, such that
	$$
	{\vphi}_{\mu, j} - {\vphi}_{j} = f_{\mu, j}(a)\,.
	$$
To prove that $f_{\mu,j}$ is identically zero, fix a value  $\bar a$
of $a$, corresponding to some point $(\bar x,\bar\xi)\in \Pi^*$ define
${\lambda}:=\|\bar a\|^{\frac{1}{\ell+1}}$,  then we have
\begin{align*}
f_{\mu, j}(a(\bar x,\bar\xi))&=f_{\mu, j}\left(\lambda^{\ell+1}a\left(\frac{\bar
  x}{ \lambda},\frac{\bar\xi}{\lambda^\ell}\right)\right)=f_{\mu,
  j}(\lambda^{\ell+1}{a({x_0},{\xi_0})})
\\
&=f_{\mu, j}\left(a\left({\lambda}
{{x_0}}, {\lambda^\ell} {{\xi_0}}\right)\right)= {\vphi}_{\mu, j}\left(
{\lambda} x_0,  {\lambda^\ell} \xi_0\right) - {\vphi}_{j} \left( {\lambda}
x_0, {\lambda^\ell} \xi_0\right)=0-0=0  .
\end{align*}
\end{proof}

\begin{corollary}
Let $\check{x}, \check{\xi}: \Pi^\circ \times \T^2\rightarrow  \Pi^* $
be the functions expressing the Cartesian coordinates in terms of the
action angle coordinates, namely s.t.
$$
(\check{x}, \check{\xi}) \circ (a, \vphi) = \textrm{Id}\,, \quad (a, \vphi) \circ (\check{x}, \check{\xi}) = \textrm{Id}\,.
$$
Then the following holds:
\begin{equation}\label{zuppa.di.cipolle}
	\check{x} (\lambda^{\ell+1}a, \vphi) = \lambda \check{x}(a, \vphi)\,, \quad \check{\xi} (\lambda^{\ell+1}a, \vphi) = \lambda^{\ell} \check{\xi}(a, \vphi) \quad \forall \lambda \in \R^+,\ (a, \vphi) \in \Pi^\circ \times \T^2\,.
\end{equation}
\end{corollary}
We now prove Lemma \ref{lem.a.dependence}:
\begin{proof}
By Lemma \ref{aaa.1}, there exists $\tilde{f} \in C^\infty(\Piz)$
such that $f = \tilde{f} \circ a$. Passing to action angle
variables on the set $\Pi^*$, one has
$$
\tilde{f}(a) ={f} (\check{x}(a, \vphi), \check{\xi}(a, \vphi)) 
$$
First we prove the estimates for $a\in\Piz$. Consider for example 
$$
\partial_{a_1} \tilde{f} = \partial_x {f} \cdot \partial_{a_1} \check x + \partial_\xi {f} \cdot \partial_{a_1} \check \xi 
$$
thus, using ${f} \in \Snoi^{m},$ and the estimates 
$$
\left|\partial_a^{\beta}\check x(a, \vphi)\right| \lesssim_{} \langle a \rangle^{\frac{1}{\ell + 1} - |\beta|}\,, \quad 	\left|\partial_a^{\beta}\check \xi(a, \vphi)\right| \lesssim_{} \langle a \rangle^{\frac{1}{\ell + 1} - |\beta|}\,.
$$
which follow from \eqref{zuppa.di.cipolle} we can estimate such a
quantity.  One gets for $a \in \Pi^\circ$
$$
| \partial_a \tilde{f}| \lesssim \langle a \rangle^{m-\delta_1
  -\frac{\ell}{\ell + 1}}  + \langle a \rangle^{m-\delta_2
  -\frac{1}{\ell + 1}}\simeq \langle a\rangle^{m-\varsigma}\,.
$$
Iterating and studying the other derivatives, one also has that $\forall \alpha \in \N^2$,
implies 
$$
| \partial_a^\alpha \tilde{f}| \lesssim \langle a
\rangle^{m-\varsigma|\alpha|}\, ,
$$
but just for $a \in \Pi^\circ$. To get a symbol defined on the whole
of $\R^2$, we consider again the cones $\cV$ and $\cC$ and the cutoff
function $\Psi$ supported in $\cV$ and equal to one in $\cC$, and just
consider the function
$$
f_c(a):=\tilde f(a)\Psi(a)(1-\chi(\left\|a\right\|))\ ,
$$
which has all the claimed properties. 
\end{proof}


\def\cprime{$'$}
\begin{thebibliography}{HSVN07}

\bibitem[Bam96]{Bam96}
Dario Bambusi.
\newblock Exponential stability of breathers in {H}amiltonian networks of
  weakly coupled oscillators.
\newblock {\em Nonlinearity}, 9(2):433--457, 1996.

\bibitem[BBM11]{BBM11}
D.~Bambusi, M.~Berti, and E.~Magistrelli.
\newblock Degenerate {KAM} theory for partial differential equations.
\newblock {\em J. Differential Equations}, 250(8):3379--3397, 2011.

\bibitem[BF16]{BF}
Larry~M Bates and Francesco Fass{\`o}.
\newblock No monodromy in the champagne bottle, or singularities of a
  superintegrable system.
\newblock {\em Journal of Geometric Mechanics}, 8(4):375, 2016.

\bibitem[BF17]{BF18}
Dario Bambusi and Alessandra Fus\`e.
\newblock Nekhoroshev theorem for perturbations of the central motion.
\newblock {\em Regul. Chaotic Dyn.}, 22(1):18--26, 2017.

\bibitem[BFS18]{BFS18}
Dario Bambusi, Alessandra Fus\`e, and Marco Sansottera.
\newblock Exponential stability in the perturbed central force problem.
\newblock {\em Regul. Chaotic Dyn.}, 23(7-8):821--841, 2018.

\bibitem[BLM20a]{BLMunbdd}
Dario Bambusi, Beatrice Langella, and Riccardo Montalto.
\newblock Growth of {S}obolev norms for unbounded perturbations of the
  {L}aplacian on flat tori.
\newblock {\em Preprint}, arXiv:2012.02654, 2020.

\bibitem[BLM20b]{BLMnr}
Dario Bambusi, Beatrice Langella, and Riccardo Montalto.
\newblock On the spectrum of the {S}chr{\"o}dinger operator on ${T}^d$: a
  normal form approach.
\newblock {\em Communications in Partial Differential Equations}, 45:1--18,
  2020.

\bibitem[BLM20c]{BLMres}
Dario Bambusi, Beatrice Langella, and Riccardo Montalto.
\newblock Spectral asymptotics of all the eigenvalues of {S}chr\"odinger
  operators on flat tori.
\newblock {\em Preprint}, arXiv:2007.07865, 2020.

\bibitem[CdV80]{CdV2}
Yves Colin~de Verdi\`ere.
\newblock Spectre conjoint d'op\'{e}rateurs pseudo-diff\'{e}rentiels qui
  commutent. {II}. {L}e cas int\'{e}grable.
\newblock {\em Math. Z.}, 171(1):51--73, 1980.

\bibitem[Cha83a]{Cha83}
Anne-Marie Charbonnel.
\newblock Calcul fonctionnel a plusieurs variables pour des operateurs
  pseudodifferentiels dans {${R}^n$}.
\newblock {\em Isra{\"e}l Journal of Mathematics}, 45(1):69--89, 1983.

\bibitem[Cha83b]{Cha83T}
Anne-Marie Charbonnel.
\newblock Spectre conjoint d'op\'{e}rateurs pseudodiff\'{e}rentiels qui
  commutent.
\newblock {\em Ann. Fac. Sci. Toulouse Math. (5)}, 5(2):109--147, 1983.

\bibitem[Cha86]{chaspe}
Anne-Marie Charbonnel.
\newblock Localisation et developpment asymptotique des elements du spectre
  conjoint d'operateurs pseudodifferentiels qui commutent.
\newblock {\em Integral Equations and operator theory}, 9(4):502--536, 1986.

\bibitem[Dui80]{Dui80}
J.~J. Duistermaat.
\newblock On global action-angle coordinates.
\newblock {\em Comm. Pure Appl. Math.}, 33(6):687--706, 1980.

\bibitem[FK04]{Fej04}
Jacques F\'{e}joz and Laurent Kaczmarek.
\newblock Sur le th\'{e}or\`eme de {B}ertrand (d'apr\`es {M}ichael {H}erman).
\newblock {\em Ergodic Theory Dynam. Systems}, 24(5):1583--1589, 2004.

\bibitem[FKT90]{FKT}
Joel Feldman, Horst Kn\"{o}rrer, and Eugene Trubowitz.
\newblock The perturbatively stable spectrum of a periodic {S}chr\"{o}dinger
  operator.
\newblock {\em Invent. Math.}, 100(2):259--300, 1990.

\bibitem[FKT91]{FKT2}
Joel Feldman, Horst Kn\"{o}rrer, and Eugene Trubowitz.
\newblock Perturbatively unstable eigenvalues of a periodic {S}chr\"{o}dinger
  operator.
\newblock {\em Comment. Math. Helv.}, 66(4):557--579, 1991.

\bibitem[Fri90]{Fri}
Leonid Friedlander.
\newblock On the spectrum of the periodic problem for the {S}chr\"{o}dinger
  operator.
\newblock {\em Comm. Partial Differential Equations}, 15(11):1631--1647, 1990.

\bibitem[HR82a]{hero2}
B.~Helffer and D.~Robert.
\newblock Asymptotique des niveaux d'\'energie pour des hamiltoniens \`a un
  degr\'e de libert\'e.
\newblock {\em Duke Math. J.}, 49(4):853--868, 1982.

\bibitem[HR82b]{HR82}
B.~Helffer and D.~Robert.
\newblock Propri\'et\'es asymptotiques du spectre d'op\'erateurs
  pseudodiff\'erentiels sur {${\bf R}^{n}$}.
\newblock {\em Comm. Partial Differential Equations}, 7(7):795--882, 1982.

\bibitem[HSVN07]{HSV}
M.~Hitrik, J.~Sj{\"o}strand, and S.~Vu~Ngoc.
\newblock Diophantine tori and spectral asymptotics for nonselfadjoint
  operators.
\newblock {\em American journal of mathematics}, 129(1):105--182, 2007.

\bibitem[Kar96]{Kar96}
Yu.~E. Karpeshina.
\newblock Perturbation series for the {S}chr\"{o}dinger operator with a
  periodic potential near planes of diffraction.
\newblock {\em Comm. Anal. Geom.}, 4(3):339--413, 1996.

\bibitem[Par08]{Par08}
Leonid Parnovski.
\newblock Bethe-{S}ommerfeld conjecture.
\newblock {\em Ann. Henri Poincar\'{e}}, 9(3):457--508, 2008.

\bibitem[PS10]{PS10}
Leonid Parnovski and Alexander~V. Sobolev.
\newblock Bethe-{S}ommerfeld conjecture for periodic operators with strong
  perturbations.
\newblock {\em Invent. Math.}, 181(3):467--540, 2010.

\bibitem[PS12]{PS12}
Leonid Parnovski and Roman Shterenberg.
\newblock Complete asymptotic expansion of the integrated density of states of
  multidimensional almost-periodic {S}chr\"{o}dinger operators.
\newblock {\em Ann. of Math. (2)}, 176(2):1039--1096, 2012.

\bibitem[Roy07]{roy}
Nicolas Roy.
\newblock A semi-classical {K}. {A}. {M}. theorem.
\newblock {\em Comm. Partial Differential Equations}, 32(4-6):745--770, 2007.

\bibitem[R{\"{u}}s01]{Rus01}
H.~R{\"{u}}ssmann.
\newblock Invariant tori in non-degenerate nearly integrable {H}amiltonian
  systems.
\newblock {\em Regul. Chaotic Dyn.}, 6(2):119--204, 2001.

\bibitem[Vel15]{Vel}
Oktay Veliev.
\newblock {\em Multidimensional periodic {S}chr\"{o}dinger operator}, volume
  263 of {\em Springer Tracts in Modern Physics}.
\newblock Springer, Cham, 2015.
\newblock Perturbation theory and applications.

\end{thebibliography}

\def\cprime{$'$}

\end{document}